\documentclass[runningheads]{lmcs}
\pdfoutput=1

\usepackage{lastpage}
\lmcsdoi{19}{2}{5}
\lmcsheading{}{\pageref{LastPage}}{}{}%
{Dec.~30,~2021}{Apr.~20,~2023}{}

\usepackage[utf8]{inputenc}
\usepackage{multicol}

\usepackage{algorithm}
\usepackage{enumitem}
\usepackage{subfig}
\usepackage{xspace}
\usepackage{wrapfig}
\usepackage{graphicx}
\usepackage{amsfonts} % for \mathbb{}
\usepackage{stmaryrd} % for [[ ]] used for semantics
\DeclareMathAlphabet{\mathpzc}{OT1}{pzc}{m}{it}
\usepackage[skip=3pt,font=footnotesize]{caption}

\usepackage{tikz}
\usepackage{tablefootnote}
\usepackage{makecell}
\usepackage{booktabs}
\usepackage{multicol}

\newcommand{\query}[1]{\textsc{#1}}
\newcommand{\mq}{\query{mq}}
\newcommand{\eq}{\query{eq}}

\renewcommand{\paragraph}[1]{\vspace{1mm}\noindent\textbf{\textit{#1}}\ }
\newcommand{\bigparagraph}[1]{\vspace{1mm}\noindent{\fontsize{11}{14}\sffamily\bfseries{#1}}\ }

\newcommand{\itemI}{\vspace{1mm}\noindent{\textcolor{black!50!white}{\fontsize{10}{10}\sffamily\bfseries{I.}}}\ }
\newcommand{\itemII}{\vspace{1mm}\noindent{\textcolor{black!50!white}{\fontsize{10}{10}\sffamily\bfseries{II.}}}\ }

\newcommand{\dmin}{\ensuremath{d_{-\infty}}}%\textit{inf}}}}
\newcommand{\dmax}{\ensuremath{d_{\infty}}}%\textit{sup}}}}

\newcommand{\class}[1]{\ensuremath{{\mathbb{#1}}}}
\newcommand{\algor}[1]{{\ensuremath{{\textbf{#1}}}}}
\newcommand{\lstar}{{\ensuremath{\algor{L}^*}}}
\newcommand{\matstar}{{\ensuremath{\algor{MAT}^*}}}

\newcommand{\method}[1]{\ensuremath{\textsf{#1}}}
\newcommand{\lang}[1]{\mathcal{L}({#1})}
\newcommand{\hatlang}[1]{\ensuremath{\hat{\mathcal{L}}({#1})}}

\newcommand{\partition}[1]{\ensuremath{\langle#1\rangle}}

\newcommand{\cpart}[1]{\ensuremath{\Pi_\textsf{conc}({#1})}}
\newcommand{\ppart}[1]{\ensuremath{\Pi_\textsf{pred}({#1})}}

\newcommand{\Asat}{\textsf{A}_\textsf{SAT}}

\newcommand{\rflinesigmaConc}{2\xspace}
\newcommand{\rflinegenGammai}{4\xspace}

\newcommand{\sat}[1]{\ensuremath{\mathit{sat}^{#1}}}
\newcommand{\size}[1]{\ensuremath{\mathit{size}^{#1}}}

\newcommand{\generalize}{{\method{Generalize}}}
\newcommand{\concretize}{{\method{Concretize}}}

\newcommand{\Generalize}{{\method{Generalize}}}
\newcommand{\Concretize}{{\method{Concretize}}}
\newcommand{\Decontaminate}{{\method{Decontaminate}}}

\newcommand{\CharDFA}{\algor{{CharDFA}}}
\newcommand{\InferDFA}{\algor{{InferDFA}}}

\newcommand{\CharSFA}{\algor{CharSFA}}
\newcommand{\InferSFA}{\algor{{InferSFA}}}
\newcommand{\Char}{\algor{Char}}
\newcommand{\Infer}{\algor{{Infer}}}

\newcommand{\algoQSFA}{\algor{Q}_{{\class{M}}}}
\newcommand{\algoQalge}{\algor{Q}_{{\alge{A}}}}
\newcommand{\pfcharsfa}{\method{CharSFA}}
\newcommand{\pfinfersfa}{\method{InferSFA}}

\newcommand{\alge}[1]{\ensuremath{\mathpzc{#1}}}
\newcommand{\aut}[1]{\ensuremath{\mathcal{#1}}}
\newcommand{\dom}[1]{\ensuremath{\mathbb{#1}}}
\newcommand{\atpreds}[1]{\ensuremath{\mathbb{#1}_0}}
\newcommand{\sema}[1]{\ensuremath{\llbracket{#1}\rrbracket}}

\newcommand{\sfa}{\textsc{SFA}}

\renewcommand{\dag}{\textsc{dag}}

\colorlet{pink}{red!40}
\colorlet{blue}{cyan!60}

\begin{document}
\title{Inferring Symbolic Automata}

\author[D. Fisman]{Dana Fisman\lmcsorcid{0000-0002-6015-4170}}[a]
\address{Ben-Gurion University, Be’er Sheva, Israel}

\author[H. Frenkel]{Hadar Frenkel\lmcsorcid{0000-0002-3566-0338}}[b]
\address{CISPA Helmholtz Center for Information Security, Saarbr{\"u}cken, Germany}

\author[S. Zilles]{Sandra Zilles\lmcsorcid{0000-0001-7834-8574}}[c]
\address{University of Regina, Regina, Canada}

\begin{abstract}

We study the learnability of \emph{symbolic finite state automata} (SFA), a model shown useful in many applications in software verification.
The state-of-the-art literature on this topic follows the \emph{query learning} paradigm,
and so far all obtained results are positive. We provide a necessary condition for 
efficient learnability of SFAs in this paradigm, from which we obtain the first negative result.

The main focus of our work lies in the learnability of SFAs under the paradigm of \emph{identification in the limit using polynomial time and data},
and its strengthening \emph{efficient identifiability}, which are concerned with the existence of a systematic set of \emph{characteristic samples}
from which a learner can correctly infer the target language.
We provide a necessary condition  for identification of SFAs in the limit using polynomial time and data, and 
a sufficient condition for efficient learnability of SFAs. 
From these conditions we derive a positive and a negative result.

The performance of a learning algorithm is typically bounded as a function of the size of the 
    representation of the target language. Since SFAs, in general, do not have a canonical form, 
    and there are trade-offs between the complexity of the predicates on the transitions and the number of transitions,
    we start by defining size measures for SFAs.
    We revisit the complexity of procedures on SFAs
    and analyze them according to these measures, paying attention to the special forms of SFAs: \emph{normalized SFAs} and \emph{neat SFAs},
    as well as to SFAs over a \emph{monotonic} effective Boolean algebra. 
    
   This is an extended version of the paper with the same title published in CSL'22~\cite{FFZ22}.
\end{abstract}

\maketitle

\section{Introduction}

\emph{Symbolic finite state automata}, SFAs for short, are an automata model in which transitions between
states correspond to predicates over a domain of concrete alphabet letters. Their purpose is to
cope with situations where the domain of concrete alphabet letters is large or infinite. As an example for automata over  finite large alphabets
consider automata over the alphabet $2^{AP}$ where ${AP}$ is a set of atomic propositions; these are used in model checking~\cite{ClarkeGP2001,BaierKatoenBook}. Another example, used in string sanitizer algorithms~\cite{HooimeijerLMSV11}, are automata over predicates on
the Unicode alphabet which consists of over a million symbols. 
An infinite alphabet is used for example in \emph{event recording automata}, a determinizable class of timed automata~\cite{AlurFH99}
in which an alphabet letter consists of both a symbol from a finite alphabet, and a non-negative real number. Formally, the transition predicates in an SFA
are defined with respect to an effective Boolean algebra as  defined in \autoref{sec:prelims}.

SFAs have proven useful in many applications~\cite{DAntoniVLM14,PredaGLM15,ArgyrosSJKK16,HuD17,SaarikiviV17,MamourasRAIK17} and consequently have been studied as a theoretical model of automata.
Many algorithms for natural questions over these automata already exist in the literature, in particular, Boolean operations, determinization, and  emptiness~\cite{DBLP:conf/icst/VeanesHT10}; minimization~\cite{DAntoniV16}; and language inclusion~\cite{KeilT14}. Recently the subject of learning automata in verification has also attracted attention, as it has been shown useful in many applications, see Vaandrager's survey~\cite{Vaandrager17}.

There already exists substantial literature on learning restricted forms of SFAs~\cite{GrinchteinJL10,MalerM14,ArgyrosSKK16,MalerM17,ChubachiDYS17},
as well as general SFAs~\cite{DrewsD17,ArgyrosD18}, and even non-deterministic residual SFAs~\cite{abs-1902-07417}.
For other types of automata over infinite alphabets, \cite{DBLP:conf/vmcai/HowarSM11} suggests learning abstractions, and~\cite{DBLP:conf/fm/Sheinvald19} presents a learning algorithm for deterministic variable automata.
All these works consider the query learning paradigm, 
and provide extensions to Angluin's \lstar\ algorithm for learning DFAs using membership and equivalence queries~\cite{Angluin87}. Unique 
to these works is the work~\cite{ArgyrosD18} which studies the learnability of SFAs taking as a parameter the learnability of the underlying algebras, providing
positive results regarding specific Boolean algebras.

{One of our contributions is to demonstrate that these positive learnability results are far from trivial. In particular, we show that there are limitations to the power of membership and equivalence queries when it comes to learning SFAs. To do so, we provide a necessary condition for efficient learnability of SFAs in the query learning paradigm, from which we obtain a negative result regarding query learning of SFAs over the propositional algebra. This is, to the best of our knowledge, the first negative result on learning SFAs with membership and equivalence queries and thus gives useful insights into the limitations of the \lstar\ framework in this context.} 

{The main focus of our work lies} on the learning paradigm of \emph{identification in the limit using polynomial time and data}.
We are 
interested in providing  sufficient or necessary conditions for  a class of SFAs to be learnable under this paradigm. To this
end, we 
show that the type of the algebra, in particular whether it is monotonic or not, 
largely influences the learnability of the class. 

Learnability of a class of languages in a certain paradigm greatly depends on the representation chosen for the language. For instance, regular languages are efficiently learnable (both in the paradigm of identification in the limit using polynomial time and data, and in the query learning paradigm using membership and equivalence queries) when represented as DFAs but not when represented as NFAs.
While we are interested in SFAs as the representations, there are various types of SFAs (with the same expressive power),
and the learnabilty results for them  may vary.

The literature on SFAs has mainly focused on a special type of SFA, termed \emph{normalized}, in which there
is at most one transition between every pair of states. This minimization of the number of transitions comes at the cost
of obtaining more complex predicates. 
We consider, in addition to normalized SFAs, another special type of SFAs that we
term \emph{neat SFAs}, which by contrast, allows several transitions between the same pair of states, but restricts
the predicates to be \emph{basic}, as formally defined in \autoref{sec:bool-alge}.

To get on the right track, we first take a global look at the complexity of the standard operations on SFAs, and how
they vary according to the special form. We revisit the results in the literature and analyze them along the measures
we find adequate for the size of an SFA: the number of states ($n$), the number of transitions ($m$) and the size of the most
complex predicate ($l$).\footnote{ Previous results have  concentrated mainly on the number of states.} 
 The results show that most procedures are more efficient on neat~SFAs.

We then turn to study identification of SFAs in the limit using polynomial time and data. We provide a necessary condition a class of SFAs $\class{M}$ should meet in order to be 
identified in the limit using polynomial time and data, and 
a sufficient condition 
a class of SFAs $\class{M}$ should meet in order to be efficiently identifiable. These conditions are expressed in terms of the existence of certain efficiently computable functions, which we 
call $\Generalize_\class{M}$,  $\Concretize_\class{M}$, and  $\Decontaminate_\class{M}$. 
We then provide positive and negative results regarding the learnability of specific classes of SFAs in this paradigm.
In particular,  
we show that the class of SFAs over \emph{any} monotonic algebra is
 efficiently identifiable. 
 
\paragraph{Comparison to the conference version}  Preliminary results of this work appear in~\cite{FFZ22}. This paper extends the results of~\cite{FFZ22} by adding a thorough discussion of the different SFA types and their effect on the complexity on different automata procedures; as well as a new theorem regarding efficient learnability, and additional examples for learning SFAs. In particular, sections~\ref{sec:types},~\ref{sec:special_forms}, and~\ref{sec:procedures}, are all new, as well as \autoref{theorem:efficient} and Examples~\ref{example:mono} and~\ref{ex:decontanimate}.

\paragraph{Outline} The rest of the paper is organized as follows. In \autoref{sec:prelims} we provide the necessary definitions on effective Boolean algebras and SFAs. Section~\ref{sec:types} introduces the special forms of SFAs.
In \autoref{sec:special_forms} we discuss transformations between the special forms.
Section~\ref{sec:procedures} then  reviews the complexity of standard automata procedures along the mentioned parameters. 

We then turn to discuss the learnability of symbolic automata. Section~\ref{sec:learning} provides a short overview and definitions regarding learnability of SFAs. In section~\ref{sec:learn_char_set} we discuss the paradigm of learnability in the limit using polynomial time and data, and provide an overview of learning DFAs in this paradigm. Sections~\ref{subsec:necessary} and~\ref{subsec:sufficient} present a necessary condition and a sufficient condition for the efficient learnability of SFAs, and sections~\ref{subsec:positive} and~\ref{subsec:negative} use these conditions to prove a positive result on the learnability of SFAs over monotonic algebras, and a negative result on the learnability of SFAs over the propositional algebra. 
Section~\ref{sec:query_learning} discusses query learning of SFAs and provides a negative result. 
We conclude in section~\ref{sec:discuss} with a short discussion.

\section{Preliminaries}\label{sec:prelims}

\subsection{Effective Boolean Algebra}\label{sec:bool-alge}
\emph{A Boolean Algebra} \alge{A} can be represented as a tuple $(\dom{D}, \class{P}, \sema{\cdot}, \bot, \top, \vee, $ 
$\wedge, \neg)$ where $\dom{D}$ is a set of domain elements; 
$\class{P}$ is a set of predicates closed under the Boolean connectives, where $\bot,\top\in\class{P}$; the component $\sema{\cdot} : \class{P}\rightarrow 2^\dom{D}$ is the so-called \emph{semantics function}. It satisfies the following three requirements: 
(i) $\sema{\bot} = \emptyset$, 
(ii) $\sema{\top} = \dom{D}$,~~and 
(iii) for all $\varphi,\psi\in \class{P}$, $~\sema{\varphi\vee \psi} = \sema{\varphi}\cup\sema{\psi}$, $~\sema{\varphi\wedge \psi} = \sema{\varphi}\cap\sema{\psi}$, and $\sema{\neg\varphi}= \dom{D}\setminus\sema{\psi}$. A Boolean Algebra is \emph{effective} if all the operations above, as well as satisfiability, are decidable.
Henceforth, we implicitly assume Boolean algebras to be~effective.

One way to define a Boolean algebra is by defining a set $\class{P}_0$ of \emph{atomic formulas} that includes $\top$ and $\bot$ and obtaining $\class{P}$ by closing $\class{P}_0$ for conjunction, disjunction and negation.
For a predicate $\psi\in\class{P}$ we say that $\psi$ is \emph{atomic} if $\psi\in\class{P}_0$. We say that 
$\psi$ is \emph{basic} if $\psi$ is a conjunction of atomic formulas.

We now introduce two Boolean algebras that are discussed extensively in the paper.

\paragraph{The Interval Algebra}%\label{ex:interval-algebra}
	is the Boolean algebra in which the domain $\dom{D}$ is the set  $\dom{Z}\cup\{-\infty,\infty\}$ of integers augmented with two special symbols with their standard semantics, and the set of atomic formulas $\class{P}_0$ consists of  intervals of the form $[a,b)$ where $a,b\in\dom{D}$. %and $a\leq b$. 
	The semantics associated with intervals is the natural one: $\sema{[a,b)}=\{z\in\dom{D}~|~ a\leq z \mbox{ and } z < b\}$. If $a\geq b$ then $\sema{[a,b)} = \emptyset$ and we have that $[a,b)$ is semantically equivalent to $\bot$. 

\paragraph{The Propositional Algebra}%\label{ex:propositional-algebra}
	is defined with respect to a set $AP=\{p_1,p_2,\ldots,p_k\}$ of atomic propositions. 
	The set of \emph{atomic predicates} $\atpreds{P}$ consists of the atomic propositions and their negations as well as $\top$ and $\bot$.
	The  domain $\dom{D}$ consists of all the possible valuations for these propositions, thus it is $\dom{B}^k$ where $\dom{B}\!=\!\{0,1\}$. The semantics of an atomic predicate $p$ is given by $\sema{p_i}=\{v\in\dom{B}^k~|~v[i]=1\}$, and similarly $\sema{\neg p_i}=\{v\in\dom{B}^k~|~v[i]=0\}$.\footnote{In this case  a basic formula is a \emph{monomial}.}

\subsubsection{Predicate Size}\label{sec:representation-and-size}
In order to reason about the complexity of operations over the Boolean algebra (and later, the efficient learnability of SFAs using such Boolean algebras), we need some measure of the size of predicates. We assume the algebra is associated with a function $\size{\class{P}}:\mathbb{P}\rightarrow\mathbb{N}$  returning for each predicate its size. 
If the algebra is defined via a set of atomic propositions,
one
can assume the existence of functions $\size{\class{P}}:\atpreds{P}\rightarrow\mathbb{N}$,
$\size{\class{P}}_{\wedge}:\mathbb{N}\times\mathbb{N}\rightarrow\mathbb{N}$,
$\size{\class{P}}_{\vee}:\mathbb{N}\times\mathbb{N}\rightarrow\mathbb{N}$,
$\size{\class{P}}_{\neg}:\mathbb{N}\rightarrow\mathbb{N}$
according to which the 
 size of predicates can be inductively computed.
Note that the size is a property of the predicate, not the set of concrete elements it represents.

\begin{exa}
For the interval algebra, we define the size of one interval to be $1$, and the size of a general predicate as the size of its parse tree, where leaves are single intervals (whose size is $1$).
Thus for example $\size{\class{P}}(([0,50)\vee[100,200))\wedge [20,60))$ is $5$ whereas the size of the semantically equivalent predicate $[20,50)$ is $1$.

Similarly, for the propositional algebra we define the size of a predicate to be the size of its parse tree. Note that Boolean functions from $\mathbb{B}^k$ to $\mathbb{B}$ can be represented in other ways as well, e.g., using Binary Decision Diagrams (BDDs)~\cite{DBLP:journals/tc/Bryant86}. This would result in a different Boolean algebra (where predicates are BDDs) with a different size measure for predicates.
\end{exa}

\subsection{Symbolic Automata} \label{sec:symbolic_automata}

A \emph{symbolic finite automaton} (\sfa) is a tuple 
$\aut{M} \! \!=\! \!(\alge{A},Q,q_\iota ,F,\Delta)$ where $\alge{A}$ 
is a Boolean algebra, $Q$ is a finite set of states, $q_\iota \in Q$ is the initial state, 
$F \subseteq Q$ is the set of final states, and $\Delta \subseteq Q \times \class{P}_\alge{A} \times Q $ is a finite set of transitions, where $\class{P}_\alge{A} $ is the set of predicates of $\alge{A}$.

We use the term \emph{letters} 
for elements of $\dom{D}$ where $\dom{D}$ is the domain of $\alge{A}$, and the term \emph{words} 
for elements of $\dom{D}^*$.
A \emph{run} of $\aut{M}$ on a word $a_1a_2\ldots a_n$ is a sequence of transitions
 $\langle q_0,\psi_1,q_1\rangle \langle q_1,\psi_2,q_2\rangle\ldots\langle q_{n-1},\psi_n,q_n\rangle$ satisfying that 
 $a_i\in \sema{\psi_i}$, that 
 $\langle q_i,\psi_{i+1},q_{i+1}\rangle\in\Delta$ and that $q_0=q_{\iota}$. 
 Such a run is said to be \emph{accepting} if $q_n\in F$. A word $w=a_1 a_2\ldots a_n$  is said to be \emph{accepted} 
  by $\aut{M}$ if there exists an accepting run of $\aut{M}$ on $w$. The set of words accepted by an SFA $\aut{M}$ is denoted 
$\lang{\aut{M}}$. We  use $\hatlang{\aut{M}}$ for the set of labeled words $\hatlang{\aut{M}}= \{ (w,1)~|~w\in\lang{\aut{M}}\}\cup\{(w,0)~|~w\notin\lang{\aut{M}}\}$.

An \sfa\ is said to be \emph{deterministic} if for every state $q\in Q$ and every letter $a\in\dom{D}$ we have that
$|\{\langle q,\psi,q' \rangle \in\Delta~|~a\in\sema{\psi}\}|\leq 1$, namely from every state and every concrete letter there exists at most one transition.
It is said to be \emph{complete} if  $|\{\langle q,\psi,q'\rangle\in\Delta~|~a\in\sema{\psi}\}|\geq 1$
for every  $q\in Q$ and  $a\in\dom{D}$, namely from every state and every concrete letter there exists at least one transition.
It is not hard to see that,  as is the case for finite 
%
%\vspace{-10mm} 
%\begin{wrapfigure}[5]{r}{0.35\textwidth}
 %   \vspace{-2mm}
    \begin{figure}[t]
    \centering
    \scalebox{1}{\small
        \begin{tikzpicture}[scale=0.11]
        \tikzstyle{every node}+=[inner sep=0pt]
        \draw [black] (18.3,-19) circle (3);
        \draw (18.3,-19) node {$q_0$};
        \draw [black] (32.4,-19) circle (3);
        \draw (32.4,-19) node {$q_1$};
        \draw [black] (32.4,-19) circle (2.4);
        \draw [black] (10.8,-19) -- (15.3,-19);
        \fill [black] (15.3,-19) -- (14.5,-18.5) -- (14.5,-19.5);
        \draw [black] (19.952,-16.519) arc (135.06693:44.93307:7.625);
        \fill [black] (30.75,-16.52) -- (30.54,-15.6) -- (29.83,-16.31);
        \draw (25.35,-13.78) node [above] {$[0,100)$};
        \draw [black] (33.36,-16.17) arc (189:-99:2.25);
        \draw (37.7,-13) node [right] {$[0,200)$};
        \fill [black] (35.23,-18.04) -- (36.1,-18.41) -- (35.94,-17.42);
        \draw [black] (18.597,-21.973) arc (33.44395:-254.55605:2.25);
        \draw (8,-22.77) node [below] {$[100,\infty)$};
        \fill [black] (16.12,-21.04) -- (15.18,-21.07) -- (15.73,-21.9);
        \draw [black] (30.128,-20.939) arc (-58.82797:-121.17203:9.231);
        \fill [black] (20.57,-20.94) -- (21,-21.78) -- (21.52,-20.93);
        \draw (25.35,-22.77) node [below] {$[200,\infty)$};
        \end{tikzpicture}}
    \caption{{The SFA $\aut{M}$ over ${\alge{A_\dom{N}}}$ 
    }}
   % \vspace{-1mm} 
    \label{fig:SFA}
    \end{figure}
%\end{wrapfigure}
%
 automata (over concrete alphabets), non-determinism does not add expressive power
but does add succinctness. When $\aut{A}$ is deterministic we use $\Delta(q,w)$ to denote the state $\aut{A}$ reaches on reading the word $w$ from state
$q$. If $\Delta(q_\iota,w)=q$ then  $w$ is  
 termed an \emph{access word to state} $q$. 
 If $w$ is the smallest access word according to lexicographic order we say that $w$ is the \emph{lex-access} word to state $q$.

\begin{exa}\label{ex:SFA}
Consider the SFA $\aut{M}$ given in \autoref{fig:SFA}. 
It is defined over the algebra ${\alge{A_\dom{N}}}$ which is the interval algebra restricted to the domain $\dom{D} = \mathbb{N}\cup \{\infty  \}$. The language of $\aut{M}$ is the set of all words over $\mathbb{D}$ 
of the form $w_1 \cdot d \cdot w_2$ where $w_1$ is some word over the domain $\mathbb{D}$, the letter $d$ satisfies $0 \leq d < 100$ and all letters of the word $w_2$ are numbers smaller than~200. 
The lex-access word to state $q_0$ is $\epsilon$, and $0$ is the lex-access word to state $q_1$.
\end{exa}

\section{Types of Symbolic Automata}\label{sec:types}
Since the complexity of a learning algorithm for a class of languages $\class{L}$ using
some representation $\class{R}$ is measured with respect to the
size of the smallest representation $R\in\class{R}$ for the unknown language $L\in\class{L}$,
we first need to agree how to measure the size of an SFA. Subsection~\ref{sec:size} explains
why the number of states is not a sufficient measure, and proposes an alternative using three parameters.
Optimizing different parameters leads to different special forms which are discussed in \autoref{sec:special-forms}.

\subsection{Size of an \sfa} \label{sec:size}
We note that there is a trade-off between the number of transitions and the
complexity of the transition predicates. 
The size of an automaton (not a symbolic one) is typically measured by its number of states. This is since for DFAs, the size of the alphabet is assumed to be a given constant, and the rest of the parameters, in particular the transition relation, are at most quadratic in the number of states. In the case of \sfa s the situation is different, as the size of the predicates labeling the transitions can vary greatly. In fact, if we measure the size of a predicate by the number of nodes in its parse \dag, then the size of a formula can grow unboundedly (and the same is true for other reasonable size measures for predicates).
The size and structure of the predicates influence the complexity of their satisfiability check, and thus the complexity of the corresponding algorithms. On the other hand there might be a trade-off between the size of the transition predicates and the number of transitions; e.g., a predicate of the form $\psi_1 \vee \psi_2 \ldots \vee \psi_k$ can be replaced by $k$ transitions, each one labeled by $\psi_i$ for $1\leq i \leq k$. 

Therefore, we measure the size of an \sfa\ by three parameters: the number of states~($n$), the maximal out-degree of a state ($m$) and the largest size of a predicate ($l$).

In addition, in order to analyze the complexity of automata algorithms discussed in \autoref{sec:boolean_operations} and \autoref{sec:special_boolean_operations}, for a class $\class{P}$ of predicates over a Boolean algebra $\alge{A}$, 
we use the complexity measure 
$\sat{\class{P}}(l)$, which is the complexity of satisfiability check for a predicate of size $l$ in $\class{P}$.  We also use $\sat{\class{P}_0}(l)$ for the respective complexities when restricted to atomic predicates.

\subsection{Special Form SFAs}\label{sec:special-forms}
We turn to define special types of \sfa s, which affect the complexity of related procedures. 

\paragraph{Neat and Normalized SFAs} 
The literature defines an \sfa\ as \emph{normalized} if for every two states $q$ and $q'$ there exists at most one transition from $q$ to $q'$.  
This definition prefers fewer transitions at the cost of potentially complicated predicates.
By contrast, preferring simple transitions at the cost of increasing the number of transitions, leads to~\emph{neat} \sfa s.
We define an \sfa\ to be \emph{neat} if
all transition predicates are basic predicates. 

\paragraph{Feasibility} 
The second distinction concerns the fact that an \sfa\ can have transitions with unsatisfiable predicates.
A symbolic automaton is said to be \emph{feasible}  if for every $\langle q,\psi,q'\rangle\in\Delta$ we have that $\sema{\psi}\neq \emptyset$. 
{Feasibility is an orthogonal property to being neat or normalized.} 

\paragraph{Monotonicity}
The third distinction we make concerning the nature of a given \sfa\ regards its underlying algebra. 
A Boolean algebra $\alge{A} $ over domain $\dom{D}$ is said to be  \emph{monotonic} if
the following conditions hold.
\begin{enumerate}
    \item There exists a total order $<$ on the elements of $\dom{D}$;
    and
    \item There exist two elements $\dmin$ and $\dmax$ such that $\dmin\leq d$ and $d \leq \dmax$ for all $d\in\dom{D}$;
    and
    \item An atomic predicate $\psi\in\atpreds{P}$ can be associated with two concrete values $a$ and $b$ such that $\sema{\psi}=\{d\in\dom{D}~:~ a\leq d < b \}$. Henceforth, we denote an atomic predicate $\psi$ over a monotonic algebra as $\psi = [a,b)$ where $\sema{\psi}=\{d\in\dom{D}~:~ a\leq d < b \}$. If $b\leq a$ then we have that $\sema{\psi} = \emptyset$ and thus the predicate is equivalent to $\bot$. 
\end{enumerate}

 The interval algebra is clearly monotonic, as is the similar algebra obtained using $\dom{R}$ (the real numbers) instead of $\dom{Z}$ (the integers). On the other hand, the propositional algebra is clearly non-monotonic.

\begin{exa}%\label{ex:SFA}
The SFA $\aut{M}$ 
from Example~\ref{ex:SFA} (\autoref{fig:SFA}) is defined over a monotonic algebra, and is neat, normalized, deterministic and complete.  
\end{exa}

\section{Transformations to Special Forms}\label{sec:special_forms}

We now address the task of transforming SFAs into their special forms as presented in \autoref{sec:types}.
We discuss transformations to the special forms  \emph{neat}, \emph{normalized} and \emph{feasible} automata, measured
as suggested using $\langle n, m,l \rangle$ --- the number of states, the maximal out-degree of a state, and the largest size of a predicate.

\subsection{Neat Automata}\label{sec:neat}
Since each  predicate in a neat SFA is a conjunction of atomic predicates,  neat automata are  intuitive, and the number of transitions in the SFA reflects the complexity of the different operations, as opposed to the situation with normalized SFAs. 
For the class $\class{P}_0$ of basic formulas, 
$\sat{\class{P}_0}(l)$  is usually more efficient than $\sat{\class{P}}(l)$, and in particular is
polynomial for the algebras we consider here. 
This is since, for a basic predicate $\varphi$ that is a conjunction of $l$ atomic predicates, satisfiability testing can be reduced to checking that  there are no two atomic predicates that contradict each other.
%This is since satisfiability testing can be reduced to checking that for a basic predicate $\varphi$ that is a conjunction of $l$ atomic predicates, there are no two atomic predicates that contradict each other.
Since satisfiability checking directly affects the complexity of various algorithms discussed in \autoref{sec:boolean_operations}, neat SFAs allow for efficient automata operations, as we show in \autoref{sec:special_boolean_operations}.

\subsubsection{Transforming to Neat} \label{sec:transforming-neat}
Given a general SFA $\aut{M}$ of size $\langle n, m,l \rangle$, we can construct a neat SFA $\aut{M}'$ of size $\langle n, m\cdot 2^l, l \rangle$, by transforming each transition predicate to a DNF
formula, and turning each disjunct  into an individual transition. 
The number of states, $n$, remains the same. However, the number of transitions can grow exponentially due to the transformation to DNF.
In the worst case, the size of the most complex  predicate can remain the same after the transformation, resulting in the same $l$ parameter for both automata.

Note that there is not necessarily a unique minimal neat SFA. For instance, a predicate $\psi$ over the propositional algebra with $AP=\{p_1,p_2,p_3\}$, satisfying $\sema{\psi}=\{[100],[101],[111]\}$ can be represented using the two basic transitions $(p_1\wedge \neg p_2)$ and $(p_1 \wedge p_2 \wedge p_3)$; or alternatively using the two basic transitions $(p_1\wedge p_3)$ and $(p_1 \wedge \neg p_2 \wedge \neg p_3)$, though it cannot be represented using one basic transition.\footnote{This is related to the fact that there is no unique DNF formula -- a neat automaton ``breaks" the DNF to transitions according to the disjunctions in the formula. If there is no unique DNF formula, then there is no unique neat SFA. }

Although in the general case, the transformation from normalized to neat SFAs is exponential, for monotonic algebras we have the following lemma, which follows directly from the definition of monotonic algebras and basic predicates.

\begin{lem}\label{obsrv:basic_monotonic}
Over a monotonic algebra,
the conjunction of two atomic predicates is also an atomic predicate; 
inductively, any basic formula that does not contain negations, over a monotonic algebra, is an atomic predicate.
In addition, the
negation of an atomic predicate is a disjunction of at most 2 atomic predicates.
\end{lem}

\begin{lem}\label{lemma:monotonic_to_neat}
Let $\aut{M}$ be a normalized SFA over a monotonic algebra $\alge{A}_{\mathit{mon}}$.
Then, transforming $\aut{M}$ into a neat SFA $\aut{M}'$ is linear in the size of $\aut{M}$.
\end{lem}

Since a DNF formula with $m$ disjunctions is a natural representation of $m$ basic transitions, 
Lemma~\ref{lemma:monotonic_to_neat} follows from the following property of monotonic algebras.

\begin{lem}\label{lemma:dnf_monotonic}
Let $\psi$ be a general formula over a monotonic algebra $\alge{A}_{\mathit{mon}}$. Then,
there exists an equivalent DNF formula $\psi_{\mathit{d}}$ of size linear in $|\psi|$.
\end{lem}

 \begin{proof}
First, we transform $\psi$ into a Negation Normal Form formula $\psi_{\mathit{NNF}}$, pushing negations inside the formula. 
When transforming to NNF, the number of atomic predicates (possibly under negation) remains the same, and so is the number of conjunctions and disjunctions. 
Since, by Lemma~\ref{obsrv:basic_monotonic}, a negation of an atomic predicate  over a monotonic algebra, namely a negation of an interval, results in at most two intervals, we get that $|\psi_{\mathit{NNF}}| \leq 2\cdot |\psi|$. Note that $\psi_{\mathit{NNF}}$ does not contain any negations, as they were applied to the intervals.
We now transform $\psi_{\mathit{NNF}}$ into a DNF formula $\psi_d$ recursively, operating on sub-formulas of $\psi_{\mathit{NNF}}$, distributing conjunctions over disjunctions.

We inductively prove that $\sema{\psi_d} = \sema{\psi_{\mathit{NNF}}}$ and $|\psi_d| \leq |\psi_{\mathit{NNF}}|$.
For the base case, if $\psi_{\mathit{NNF}}$ is a single interval $[a,b)$, then $[a,b)$ is in DNF and we are done.

For the induction step, consider the following two cases.

\begin{enumerate}
    \item Assume $\psi_{\mathit{NNF}} = \psi_1 \vee \psi_2$. By the induction hypothesis, there exists DNF formulas $\psi_{1d}$ and $\psi_{2d}$ such that $
    \sema{\psi_{id}} = \sema{\psi_i}$ and $|\psi_{id}| \leq |\psi_i|$ for $i=1,2$.
    Then, $\psi_d = \psi_{1d}\vee \psi_{2d}$ is equivalent to $\psi_{\mathit{NNF}}$ and at most of the same size. 
    \item Assume $\psi_{\mathit{NNF}} = \psi_1 \wedge \psi_2$. Again, by the induction hypothesis, instead of $\psi_1\wedge \psi_2$ we can consider $\psi_{1d}\wedge\psi_{2d}$ where $\psi_{1d}$ and $\psi_{2d}$ are in DNF. That is $\psi_{1d} = \bigvee_{i=1}^k [a_i, b_i)$ and $\psi_{2d} = \bigvee_{j=1}^l[c_j, d_j)$. 
    Then, we have the following:
    
    $$\psi_{1d} \wedge\psi_{2d} = \left( \bigvee_{i=1}^k [a_i, b_i) \right) \wedge \left( \bigvee_{j=1}^l[c_j, d_j) \right) = \bigvee_{i=1}^{k}\bigvee_{j=1}^{l} \Big(  [a_i, b_i)\wedge[c_j, d_j) \Big)$$
    
    \noindent From properties of intervals, each conjunction $[a_i, b_i)\wedge[c_j, d_j)$ is of the form
 $[\max\{a_i, c_j\},\allowbreak \min\{b_i, d_j\})$. The intervals in $\{ [a_i, b_i) : 1\leq i\leq k \}$ do not intersect (otherwise it would have resulted in a longer single interval), and the same for $\{[c_j, d_j) : 1\leq j\leq l\}$. Thus, every element $a_i$ or $c_j$ can define at most one interval of the form $[\max\{a_i, c_j\}, \min\{b_i, d_j\})$. That is, the DNF formula $\psi_d = \bigvee_{i=1}^{k}\bigvee_{j=1}^{l} \Big([a_i, b_i)\wedge[c_j, d_j)\Big)$ contains at most $k+l$ intervals, as the others are empty intervals. Since the size of the original $\psi_{NNF}$ is $k+l$, we have that $|\psi_d|\leq|\psi_{NNF}|$.
\end{enumerate}
To conclude, since $\psi_{NNF}$ is linear in the size of $\psi$, and the size of $\psi_d$ is at most the size of $\psi_{NNF}$, we have that the translation of $\psi$ into the DNF formula $\psi_d$ is linear.  
 \end{proof}

\subsection{Normalized Automata}
Neat automata stand in contrast to normalized ones. In a normalized SFA, there is at most one transition between every pair of states, which  allows for a succinct formulation of the condition to transit from one state to another.
On the other hand, this makes the predicates on the transitions {structurally} more complicated. Given a general SFA $\aut{M}$ with parameters $\langle n, m,l \rangle$, we can easily construct a normalized SFA $\aut{M}'$ as follows. For every pair of states $q$ and $q'$, construct a single edge labeled with the predicate $\bigvee _{\langle q, \varphi, q' \rangle \in\delta} \varphi$. Then, $\aut{M}'$ has size  $\langle n , \min (n^2,m),  \size{\class{P}}_{\vee^m}(l)  \rangle$, where we use $\size{\class{P}}_{\vee^m}(l)$ to denote the size of $m$ disjunctions of predicates of size at most $l$. 

Note that there is no unique minimal normalized automaton either, since in general a Boolean formula can have two semantically equivalent, yet syntactically different expressions in the underlying representation system, e.g., two distinct BDDs can represent the same formula. 
 However, in \autoref{sec:special_boolean_operations} we show that over monotonic algebras there is a  canonical minimal normalized SFA.

The complexity of $\sat{\class{P}}(l)$ for general formulas (corresponding to normalized SFAs) is usually exponentially higher than for basic predicates (and thus for neat SFAs). 
In addition, as we show above, generating a normalized automaton is an easy operation. This motivates working with neat automata, and generating normalized automata as a last step, if desired (e.g., for presenting a graphical depiction of the automaton). 

\subsection{Feasible Automata}\label{sec:feasible_automata}
The motivation for feasible automata is clear; if the automaton contains unsatisfiable transitions, then its size is larger than necessary, and the redundancy of transitions makes it less interpretable. Thus, infeasible \sfa s add complexity both algorithmically and for the user, as they are more difficult to understand. In order to generate a feasible SFA from a given SFA $\aut{M}$, we need to traverse the transitions of $\aut{M}$ and  test the satisfiability of each transition. 
The parameters $\langle n,m,l   \rangle$ of the SFA remain the same since 
there is no change in the set of states, and there might be no change in transitions as well (if they are all satisfiable). 

In the following, we usually assume that the automata are feasible, and when applying algorithms, we require the output to be feasible as well.

\section{Complexity of standard automata procedures on SFAs}\label{sec:procedures}
In this section we analyze the complexity of automata procedures on SFAs, in terms
of their effect on the parameters $\langle n,m, l\rangle$. We start in \autoref{sec:boolean_operations} with
examining general SFAs, and then in \autoref{sec:special_boolean_operations} discuss the effects on special SFAs.

\subsection{Complexity of Automata Procedures for General SFAs}\label{sec:boolean_operations}
We turn to discuss Boolean operations, determinization and minimization, and decision procedures (such as emptiness and equivalence) for the different types of SFAs. For intersection and union, the product construction of SFAs was studied in~\cite{DBLP:conf/icst/VeanesHT10,DBLP:conf/vmcai/HooimeijerV11}. There, the authors assume a normalized SFAs as an input, and do not delve into 
the effect of the construction on the number of transitions and the complexity of the resulting predicates. Determinization of SFAs was studied in~\cite{DBLP:conf/icst/VeanesHT10}, and~\cite{DBLP:conf/popl/DAntoniV14} study minimization of SFAs,  
assuming the given SFA is~normalized.

Table~\ref{table:operations} shows the sizes of the SFAs resulting from the mentioned operations, in terms of $\langle n, m,l \rangle$. The analysis applies to all types of SFAs, not just normalized ones.
The time complexity for each operation is given in terms of the parameters $\langle n,m,l  \rangle$ and the complexity of feasibility tests for the resulting SFA, as discussed in \autoref{sec:feasible_automata}. 
Table~\ref{table:decision} summarizes the time complexity of decision procedures for SFAs: emptiness, inclusion, and membership. Again, the analysis applies to all types of SFAs. %The analysis of Tables~\ref{table:operations} and~\ref{table:decision} is explained in the remainder of this subsection. 
We note that in many applications of learning in verification, the challenging part is implementing the teacher (e.g., in~\cite{DBLP:journals/fmsd/PasareanuGBCB08,ChocklerKKS20,DBLP:conf/tacas/FrenkelGPS20,DBLP:journals/sttt/FrenkelGPS22}). In
such cases the complexity of  membership and equivalence queries as well as standard automata operations plays a major role.

In both tables we consider two SFAs  $\aut{M}_1$ and $\aut{M}_2$ with parameters $\langle n_i, m_i, l_i \rangle$ for $i=1,2$, over algebra $\alge{A}$ with predicates $\class{P}$. 
We use $\size{\class{P}}_{\wedge^m}(l)$ for an upper bound on the size of $m$ conjunctions of predicates of size at most $l$. 
All SFAs are assumed to be deterministic, except of course for the input for determinization.

\begin{table}[t]
\centering
\begin{tabular}{|c|c|}
\hline
\textbf{Operation} & $\mathbf{\langle n, m, l \rangle } $\\ \hline 
  product construction $\aut{M}_1$, $\aut{M}_2$  &  $\langle n_1 \times n_2,\ m_1 \times m_2,\ \size{\class{P}}_{\wedge}(l_1, l_2) \rangle$

      \\[1mm]
     complementation of deterministic $\aut{M}_1$\tablefootnote{For complementation, no feasibility check is needed, since we assume a feasible input.
     }    &  $\langle n_1+1 ,\ m_1+1,\ \size{\class{P}}_{\neg}(\size{\class{P}}_{\vee^{m_1}}(l_1)) \rangle$
     \\[1mm]
     determinization of $\aut{M}_1$    &  $\langle 2^{n_1},\ 2^{m_1},\ \size{\class{P}}_{\wedge^{n_1\times m_1}}(l_1)
 \rangle$ \tablefootnote{To determinize transitions, conjunction may be applied $n_1\times m_1$ times, according to the number of states that correspond to a new deterministic state.
 }
 \\[1mm]
  minimization of $\aut{M}_1$    &  $\langle n_1, m_1,\ \size{\class{P}}_{\wedge^{m_1}}(l_1)
 \rangle$
     
     \\ \hline
\end{tabular}
  
\caption{Analysis of standard automata procedures on SFAs. 
}
\label{table:operations}
\end{table}

\begin{table}[t]
\centering
\begin{tabular}{|c|c|}  \hline 
\textbf{Decision Procedures} & \textbf{Time Complexity}\\ 
\hline  
  emptiness &  linear in $n, m$  
  \\[1mm]
  emptiness + feasibility &   $n\times m \times \sat{\class{P}}(l)$
  \\[1mm]

  membership of $\gamma_1 \cdots \gamma_t\in\mathbb{D}^*$ & 
  $\sum_{i=1}^t  \sat{\class{P}}(         \size{\class{P}}_{\wedge}(l, |\psi_{\gamma_i}|)) $ \tablefootnote{Where $\psi_{\gamma_i}$ is a predicate describing $\gamma_i$.}
  \\[1mm]
  inclusion $\aut{M}_1\subseteq\aut{M}_2$ & \makecell{
  $(n_1\times n_2) \times (m_1\times m_2 )\times \sat{\class{P}} (\size{\class{P}}_{\wedge}(l_1, l_2))$ 
  }
  \\[1mm]
 \hline
\end{tabular}
 
\caption{Analysis of time complexity of decision procedures for SFAs.
}
\label{table:decision}

\end{table}

We now briefly describe the algorithms we analyze in both tables.

\paragraph{Product Construction~\cite{DBLP:conf/icst/VeanesHT10,DBLP:conf/vmcai/HooimeijerV11}}  
The product construction for SFAs is similar to the product of DFAs --- the set of states is the product of the states of $\aut{M}_1$ and $\aut{M}_2$; and a transition is a synchronization of transitions of $\aut{M}_1$ and $\aut{M}_2$. That is, a transition from $\langle q_1, q_2 \rangle$ to $\langle p_1, p_2 \rangle$ 
can be made while reading a concrete letter $\gamma$, 
iff $\langle q_1, \psi_1, p_1 \rangle\in \Delta_1$ and $\langle q_2, \psi_2, p_2 \rangle\in \Delta_2$ and $\gamma$ satisfies both $\psi_1$ and $\psi_2$. Therefore, the predicates labeling transitions in the product construction are conjunctions of predicates from the two SFAs $\aut{M}_1$ and~$\aut{M}_2$.

\paragraph{Complementation} 
In order to complement a deterministic SFA $\aut{M}_1$, we first need to make $\aut{M}_1$ complete. In order to do so, we add one state which is a non-accepting sink, and from each state we add at most one transition which is the negation of all other transitions from that state.   
If $\aut{M}_1$ is complete, then complementation simply switches accepting and non-accepting states, resulting in the same parameters $\langle n_1, m_1, l_1 \rangle$.

\paragraph{Determinization~\cite{DBLP:conf/icst/VeanesHT10}} 
In order to make an SFA deterministic, the algorithm of~\cite{DBLP:conf/icst/VeanesHT10} uses the subset construction for DFAs, resulting in an exponential blowup in the number of states. However, in the case of SFAs this is not enough, and the predicates require special care. Let
 $P = \{ q_1, \cdots, q_t \}$ 
 be a state in the deterministic SFA, where $q_1, \ldots, q_t$ are states of the original SFA $\aut{M}_1$, and let $\psi_1, \ldots, \psi_t$ be some predicates labelling outgoing transitions from $q_1, \ldots q_t$, correspondingly. 
 Then, in order to determinize transitions, the algorithm of \cite{DBLP:conf/icst/VeanesHT10} 
computes the conjunction $\bigwedge_{i=1}^t \psi_i$, which labels a single transition from the state $P$.

\paragraph{Minimization~\cite{DBLP:conf/popl/DAntoniV14}} 
Given a deterministic SFA $\aut{M}_1$, the output of minimization is an equivalent deterministic SFA with a minimal number of states. When constructing such an SFA, the number of states and transitions cannot grow. However, as in determinization, if two states of $\aut{M}_1$ are replaced with one state, then outgoing transitions might overlap, resulting in a non-deterministic SFA. 
D'Antoni and Veanes~\cite{DBLP:conf/popl/DAntoniV14} suggest several algorithms to cope with this difficulty. One of their approaches is to compute \emph{minterms}, which are the smallest conjunctions of outgoing transitions. Minterms then do not intersect, and thus the output is deterministic. 
Their other approaches avoid computing minterms, but are able to achieve the same goal.

\paragraph{Emptiness}  
If we assume a feasible SFA $\aut{M}$ as an input, then in order to check for emptiness we need to find an accepting state which is reachable from the initial state (as in DFAs). If we do not assume a feasible input, we need to test the satisfiability of each transition, thus the complexity depends on the complexity measure $\sat{\class{P}}(l)$.

\paragraph{Membership}  
Similarly to emptiness, in order to check if a concrete word $\gamma_1 \cdots \gamma_n$ is in $\aut{L}(\aut{M})$, we need
not only check if it reaches an accepting state but also
locally consider the satisfiability of each transition. In the case of membership, we need to check whether the letter $\gamma_i$ satisfies the predicate on the corresponding transition.
  
\paragraph{Inclusion}  
Deciding inclusion amounts to checking emptiness and feasibility of $\aut{M}_1  \cap \overline{\aut{M}_2}$. We assume here that both $\aut{M}_1$ and $\aut{M}_2$ are deterministic and complete.

\subsection{Complexity of Automata Procedures for Special SFAs}\label{sec:special_boolean_operations}
We now discuss the advantages of neat SFAs and of monotonic algebras, in the context of the algorithms presented in the tables, and show that, in general, they are more efficient to handle compared to other SFA types.

\subsubsection{Neat SFAs}
 As can be observed from Table~\ref{table:decision}, almost all decision procedures regarding SFAs depend on $\sat{\class{P}}(l)$.
For neat SFAs it is more precise to say that they depend on $\sat{\class{P}_0}(l)$, 
namely on the satisfiability of atomic predicates rather than arbitrary predicates.
Since $\sat{\class{P}_0}(l)$ is usually less costly than $\sat{\class{P}}(l)$, most decision procedures are more efficient 
on neat automata.
Here, we claim that applying automata algorithms on neat SFAs preserves their neatness, thus suggesting that neat SFAs may be preferable in many applications.

\begin{lem}
Let $\aut{M}_1$ and $\aut{M}_2$ be neat SFAs. Then the algorithms for their product construction, complementation, determinization and minimization
 discussed in \autoref{sec:boolean_operations} result in a neat SFA.
\end{lem}

\begin{proof}
Observing the procedures for product construction~\cite{DBLP:conf/icst/VeanesHT10,DBLP:conf/vmcai/HooimeijerV11}, determinization~\cite{DBLP:conf/icst/VeanesHT10} and minimization~\cite{DBLP:conf/popl/DAntoniV14} constructions, one can see that they use only conjunctions in order to construct the predicates on the output SFAs. Thus, if the predicates on the input SFAs are basic, then so are the output predicates. 
\end{proof}

\subsubsection{Monotonic Algebras}
We now consider the class $\class{M}_{\alge{A}_{mon}}$ of SFAs over a monotonic algebra $\alge{A}_{mon}$ with predicates $\class{P}$.
We first discuss $\size{\class{P}}_{\wedge}(l_1, l_2)$ and  $\sat{\class{P}}(l)$, as they are essential measures in automata operations.
Then we show that for $\aut{M}_1$ and $\aut{M}_2$ in the class $\class{M}_{\alge{A}_{mon}}$, the product construction is linear in the number of transitions, adding to the efficiency of SFAs over monotonic algebras.

\begin{lem}\label{lem:conjunction-in-monotonic}
    Let $\psi_1$ and $\psi_2$ be formulas over a monotonic algebra $\alge{A}_{mon}$. Then $\size{\class{P}}_{\wedge}(|\psi_1|,$ $ |\psi_2|)$ is linear in $|\psi_1|+|\psi_2|$ and  $\sat{\class{P}}(|\psi_1|)$ is linear in $|\psi_1|$. 
\end{lem}

\begin{proof}
Transforming to DNF is linear, as follows from \autoref{lemma:dnf_monotonic}. There, we show that the conjunction of two DNF formulas of sizes $k$ and $l$ has size $k+l$, which implies that the conjunction of general formulas has linear size.
In addition, $\sat{\class{P}}(l)$ is trivial for a single interval,
and following Lemma~\ref{lemma:dnf_monotonic}, is linear for general formulas.
For an interval $[a,b)$, satisfiability checking  amounts to the question ``is $a<b$?''. 
\end{proof}
 
\begin{lem}\label{lemma:product monotonic}
    Let $\aut{M}_1$ and $\aut{M}_2$ be deterministic SFAs over a monotonic algebra $\alge{A}_{mon}$. Then the out-degree of their product SFA $\aut{M}$ is at most $m = 2\cdot (m_1 + m_2)$.
\end{lem}
 
 \begin{proof}
From Lemma~\ref{lemma:monotonic_to_neat} and Lemma~\ref{lemma:dnf_monotonic}, 
we can construct neat SFAs $\aut{M}'_1$ and $\aut{M}'_2$ of sizes $\langle n_i, 2m_i,$ $ l_i\rangle$ for $i\in\{1,2\}$, that have the same languages as $\aut{M}_1$ and $\aut{M}_2$, respectively.
Similarly to the proof of Lemma~\ref{lemma:dnf_monotonic}, 
each transition $\langle \langle q_1, q_2 \rangle,  [a,b)\wedge [c,d) , \langle p_1, p_2 \rangle \rangle$ in the product SFA results in a predicate $[\max\{a,c\},\allowbreak \min\{b,d\})$. Then, for $q_1\in Q_1$, every minimal element in the set of $q_1$'s outgoing transitions can define at most one transition in $\aut{M}$, and the same for a state $q_2\in Q_2$, and so the number of transitions from $\langle q_1, q_2 \rangle$ is at most $m_1 + m_2$, as required. 
\end{proof}

\begin{lem}\label{lemma:complete_monotonic}
Let $\aut{M}$ be a neat SFA over a monotonic algebra. Then, transforming $\aut{M}$ into a complete SFA $\aut{M'}$ is polynomial in the size of $\aut{M}$.
\end{lem}

\begin{proof}
In order to complete $\aut{M}$, we add a non-accepting sink $r$ in case it does not already exist, and at most $m+1$ transitions from each state $q$ to $r$, when $m$ is the out-degree of the SFA:
Let $[ a, b )$ and $[ c, d )$ be two predicates labeling outgoing transitions of $q$, where
$c$ is the minimal left end-point of a predicate such that
$b<c$. Then, in order to complete~$\aut{M}$, we need to add a transition to the sink, labeled by the predicate $[b, c)$. In addition, for the predicate $[a, b)$ where there is no $c>b$ that defines another predicate, if $b\neq \dmax$ we add $[b, \dmax)$, and similarly we add $[\dmin, a)$. Then, for each state we add at most $m+1$ new transitions, 
resulting in at most $|Q|\times (m+1)$ new transitions. 
\end{proof}

\begin{defi}\label{def:canon}
For predicates over a monotonic algebra, we define a \emph{canonical representation} of a predicate $\psi$ as the simplified DNF formula which is the disjunction of all maximal disjoint intervals satisfying $\psi$. 
\end{defi}

Note that every predicate $\psi$ over a monotonic algebra 
defines a unique partition of the domain into maximal disjoint intervals. This unique partition corresponds to a simplified DNF formula, which is exactly the canonical representation of $\psi$. 

\begin{exa}
The canonical representation of $\psi = [0, 100) \wedge( [50, 150) \vee [20, 40) )$ is $[20, 40) \vee [50, 100)$.
\end{exa}

\begin{lem}\label{lemma:canonical}
Let $\aut{M}$ be an SFA over a monotonic algebra. Then: 
\begin{enumerate}
    \item There is a unique minimal-state neat SFA $\aut{M'}$ such that $\aut{L}(\aut{M})= \aut{L}(\aut{M'})$. \label{item:unique} 
    \item There is a canonical minimal-state normalized SFA $\aut{M''}$ such that $\aut{L}(\aut{M})= \aut{L}(\aut{M''})$. \label{item:canonic} 
\end{enumerate}
\end{lem}

\begin{proof}
For a language $\aut{L} = \aut{L}(\aut{M})$ for some SFA $\aut{M}$,
the minimal number of states in an SFA corresponds,  similarly to DFAs, to the number of equivalence classes in the
equivalence relation $N$ defined by $(u,v)\in N \Longleftrightarrow \forall z\in \dom{D}^*: (uz\in \aut{L} \Leftrightarrow vz\in\aut{L}) $~\cite{Myhill57,Nerode58}.
Indeed if $(u,v)\in N$ then there is no reason that reading them (from the initial state) should end up in different states, and if $(u,v)\notin N$ then
reading them (from the initial state) must lead to different states.

As for transitions, we have the following. 

\begin{enumerate}
    \item Let $\psi$ be a general predicate labeling a transition in $\aut{M}$. Then $\psi$ defines a unique partition of the domain into maximal
    disjoint intervals, which are exactly the transitions in a neat SFA. Then, the minimal state neat SFA is unique, where its transitions correspond exactly to these maximal disjoint intervals. 
    \item For normalized transitions, we can use Lemma~\ref{lemma:dnf_monotonic} to transform a general predicate labeling a transition to a DNF predicate in linear time. A DNF predicate over a monotonic algebra is in-fact a disjunction of disjoint intervals, where the construction of Lemma~\ref{lemma:dnf_monotonic} obtains the maximal disjoint intervals. 
    Then, to obtain a canonical representation, we order these intervals by order of their minimal elements.  \qedhere
\end{enumerate}
\end{proof}

\section{Learning SFAs}\label{sec:learning}

We turn to discuss the learnability of symbolic automata.
In grammatical inference, loosely speaking, we are interested in learning a class of languages $\class{L}$ over
an alphabet $\Sigma$, from examples which are words over $\Sigma$. Examples for classes of languages can be the
set of regular languages, the set of context-free languages, etc. A learning algorithm, aka a \emph{learner}, is expected
to output some concise representation of the language from a class of representations~$\class{R}$ for the class $\class{L}$.
For instance, in learning the class $\class{L}_{reg}$ of regular languages one might consider the class $\class{R}_\textsc{dfa}$ of DFAs,
or the class $\class{R}_\textsc{lin}$ of right linear grammars, since both are capable of expressing all regular languages.\footnote{The 
	class of regular languages was shown learnable via various representations including DFAs~\cite{Angluin87}, NFAs~\cite{BolligHKL09}, and AFAs (alternating finite automata)~\cite{AngluinEF15}.}  
We often say that a class of representations $\class{R}$ is learnable (or not) when we mean that class of languages $\class{L}$ is learnable (or not) via the class of representations~$\class{R}$.
The complexity of learning an unknown language $L\in\class{L}$ via $\class{R}$ is typically measured with respect to the size of the smallest representation $R_L\in\class{R}$ for~$L$.
For instance, when learning $\class{L}_{reg}$ via $\class{R}_{\textsc{dfa}}$ a learner is expected to output a DFA for an unknown language in time that is polynomial in the number 
of states of the minimal DFA for~$L$. 

In our setting we are interested in learning regular languages using as a representation a class of SFAs over a certain algebra. To measure complexity we must agree on how to measure the size of an SFA. Thus, as discussed in \autoref{sec:size} we represent the size of an SFA using the parameters $\langle n, m, l \rangle$  of the number of states of the SFA, the number of transitions, and the size of the largest  predicate. 
Another important factor  regarding size and canonical forms of SFAs, is the underlying algebra, specifically,
whether it is monotonic or~not.

\paragraph{Learning Paradigms}
The exact definition regarding learnability of a class depends on the \emph{learning paradigm}. 
In this
work we consider two widely studied paradigms: \emph{identification in the limit using polynomial time and data} and \emph{learning with membership and equivalence queries}. 
Their definitions are provided in sections~\ref{sec:learn_char_set} and~\ref{sec:query_learning}, respectively. 
Note that in general, a positive or negative result in one paradigm, does not imply the same result in another paradigm. We discuss this further in \autoref{sec:discuss}.

\paragraph{Basic SFAs} 
To provide results regarding the learnability of SFAs, we study classes of SFAs that contain all basic SFAs, defined as follows. 

\begin{defi}\label{def:non-trivial}
An SFA $\aut{M}$ over a Boolean Algebra $\alge{A}$ with a set of predicates $\class{P}$ is termed \emph{basic} if it is of the form 
$\aut{M}_\varphi=(\alge{A},\{q_\iota,q_{ac}, q_{rj}\},q_\iota ,\{q_{ac}\},\Delta)$ where $\varphi\in\class{P}$ and $\Delta=\{\langle q_\iota,\varphi,q_{ac}\rangle, \langle q_\iota,\neg \varphi,q_{rj}\rangle , \langle q_{rj},\top,q_{rj}\rangle, \langle q_{ac},\top,q_{rj}\rangle   \}$.
Note that $\aut{M}_\varphi$ 
accepts only  words of length one
consisting of a concrete letter satisfying $\varphi$, and it is minimal among all complete deterministic SFAs accepting this language (minimal in both number of states and number of transitions). 
\end{defi}

In the sequel, our results are regarding classes of SFAs that contain all basic SFAs $\aut{M}_\varphi$ for all $\varphi\in \class{P}$.

\section{Efficient Identifiability}\label{sec:learn_char_set}%Learning from a Characteristic Sample Set} 
%\input{j_identification}

%\subsection{Motivation,Definition,Reg Langs}

While in the better-known setting of \emph{active learning} (namely, query learning with \mq s and \eq s) the learner can select any word and query about its membership in the unknown language, in \emph{passive learning} the 
learner is given a set of words, and for each word $w$ in the set, a label $b_w$ indicating whether $w$ is in the unknown language or not.
Formally, a \emph{sample} for a language~$L$ is a finite set~$\mathcal{S}$ consisting of labeled examples, that is, pairs of the form $\langle w, b_w  \rangle $ where $w$ is a word and $b_w\in \{ 0,1 \}$ is its label, satisfying that $b_w = 1$ if and only if $w\in L$. The words that are labeled~$1$ are termed \emph{positive} words, and those that are labeled~$0$ are termed \emph{negative} words. Note that if $L$ is recognized by an automaton~$\aut{M}$, we have that $\mathcal{S}\subseteq\hatlang{\aut{M}}$ (as defined in \autoref{sec:symbolic_automata}). If~$\mathcal{S}$ is a sample for~$L$ we often say that~$\mathcal{S}$ \emph{agrees with}~$L$. 
 Given two words $w, w'$, we say that $w$ and $w'$ are \emph{not equivalent} with respect to~$\mathcal{S}$, denoted
    $w \not\sim_\aut{S} w'$, iff there exists $z$ such that $\langle w z, b \rangle, \langle w' z, b' \rangle \in \aut{S}$ and $b\neq b'$. Otherwise we say that $w$ and $w'$ are equivalent  with respect to $\mathcal{S}$, and denote $w\sim_\aut{S} w'$.

Given a sample $\aut{S}$ for a language $L$ over a concrete domain $\dom{D}$, it is possible to construct a DFA that agrees with $\aut{S}$
in polynomial time. Indeed one can create  the \emph{prefix-tree automaton}, a simple automaton that accepts all and only the positively labeled words in the sample. Clearly the constructed automaton may not be the minimal automaton that agrees with~$\aut{S}$. There are several algorithms that infer such a minimal automaton, in particular the popular RPNI~\cite{RPNI},
that merges the states of the prefix-tree automaton and results in an automaton that may accept an infinite language.
 Obviously though, this procedure is not guaranteed to return an automaton for the unknown language, as the sample may not provide sufficient information. For instance if $L=aL_1 \cup bL_2$ and the sample contains only words starting with $a$, there is no way for the learner to infer $L_2$ and hence also~$L$ correctly. 
One may thus ask, given a language~$L$, what should a sample contain in order for a passive learning algorithm to infer $L$ correctly, and can such sample be of polynomial size with respect to a minimal representation (e.g., a DFA) for the language.

One approach to answer these questions is captured in the paradigm of \emph{identification in the limit using polynomial time and data}. This model was proposed by Gold~\cite{DBLP:journals/iandc/Gold78}, who also showed that it admits learning of regular languages represented by DFAs. We follow de la Higuera's more general definition~\cite{Higuera97}.\footnote{This paradigm may seem related to conformance testing. The relation between conformance testing for Mealy machines and automata learning of DFAs has been explored in~\cite{DBLP:conf/fase/BergGJLRS05}.} 
This definition requires that for any language~$L$ in a class of languages $\class{L}$ represented by $\class{R}$, there exists a sample $\aut{S}_L$ of size polynomial in the size of the smallest representation $R\in\class{R}$ of $L$ (e.g., the smallest DFA for $L$),  such that a valid learner can infer the unknown language $L$ from the information contained in $\aut{S}_L$. The set $\aut{S}_L$ is then termed a \emph{characteristic sample}.\footnote{De la Higuera's notion of characteristic sample is a core concept in grammatical inference, for various reasons. Firstly, it addresses shortcomings of several other attempts to formulate polynomial-time learning in the limit \cite{Angluin88,Pitt89}. Secondly, this notion has inspired the design of popular algorithms for learning formal languages such as, for example, the RPNI algorithm~\cite{RPNI}. Thirdly, it was shown to bear strong relations to a classical notion of machine teaching~\cite{GoldmanM96}; models of the latter kind are currently experiencing increased attention in the machine learning community~\cite{Zhu18}.} 
Here, a valid learner is an algorithm that learns the target language exactly and efficiently. In particular, a valid learner produces in polynomial time a representation that agrees with the provided sample. The learner also has to correctly learn the unknown language $L$ when given the characteristic sample $\aut{S}_L$ as input. Moreover, if the input sample $\aut{S}$ subsumes $\aut{S}_L$ yet is still consistent with $L$, the additional information in the sample should not ``confuse'' the learner; the latter still has to output  a correct representation for $L$.  (Intuitively, this requirement precludes situations in which the sample consists of some smart encoding of the representation that the learner simply deciphers. In particular, the learner will not be confused if an adversary ``contaminates'' the characteristic sample by adding labeled examples for the target language.)  We provide the formal definition after the following informal example.

\begin{exa}
For the class of DFAs, let us consider the regular language $L = a^*$ over the alphabet $\{a, b \}$. Further, consider a sample set $\mathcal{S} = \{ \langle  \epsilon, 1 \rangle,  \langle  a , 1 \rangle , \langle b, 0 \rangle, \langle  bb, 0 \rangle, \langle  ba, 0 \rangle \} $ for $L$. There is a valid learner for the class of all DFAs that uses the sample $\mathcal{S}$ as a characteristic sample for $L$. By definition, such a learner has to output a DFA for $L$ when fed with $\mathcal{S}$, but also has to output equivalent DFAs whenever given any superset of $\mathcal{S}$ as input, as long as this superset agrees with $L$. Naturally, the sample $\mathcal{S}$ is also consistent with the regular language $L'=\{\epsilon, a\}$. However, this does not pose any problem, since the same learner can use a characteristic sample for $L'$ that disagrees with $L$, for example, $\mathcal{S}'=\{\langle\epsilon,1\rangle,\langle a,1\rangle,\langle aa,0\rangle\}$. When defining a system of characteristic samples like that, the core requirement is that the size of a sample be bounded from above by a function that is polynomial in the size of the smallest DFA for the respective target language.
\end{exa}

\begin{defi}[identification in the limit using polynomial time and data~\cite{Higuera97}]\label{def:ident-limit}
    A~class of languages $\class{L}$ is said to be 
    \emph{identified in the limit using polynomial time and data} via representations in a class $\class{R}$ 
    if there exists a learning algorithm $\algor{A}$ such that the following two requirements are met. 
    \begin{enumerate}
        \item
        Given a finite sample $\mathcal{S}$ of labeled examples, $\algor{A}$ returns a hypothesis $\aut{R}\in\class{R}$ that agrees with $\mathcal{S}$ in polynomial time. 
        \item
        For every language $L\in\class{L}$, there exists a sample $\aut{S}_L$, termed \emph{a characteristic~sample}, of size polynomial in the minimal representation $\aut{R}\in\class{R}$ for $L$ such that the algorithm $\algor{A}$ returns a correct hypothesis
        when run on any sample $\mathcal{S}$ for $L$ that subsumes~$\mathcal{S}_L$.
    \end{enumerate}
\end{defi}
Note that the first condition ensures polynomial time and the second polynomial data. However, the latter is not a worst-case measure; the algorithm may fail to return a correct hypothesis on arbitrarily large finite samples (if they do not subsume a characteristic set). 

Note also that the definition does not require the existence of an efficient algorithm that constructs a characteristic sample for each language in the underlying class.
When such an algorithm is also available we say that the class is \emph{efficiently identifiable}. 
The following result shows that efficient identifiability does not trivially follow from identifiability; in fact it makes the much stronger statement that not even \emph{computability} of characteristic sets follows from their existence.

\begin{thm}\label{theorem:efficient}
There exists a class of languages that possesses polynomial-size characteristic sets, yet without the ability to construct such sets effectively. 
\end{thm}

 We first prove Theorem~\ref{theorem:efficient} and then provide a definition for efficient identification. 

\begin{proof}
We present a class of recursive languages that is identifiable in the limit from polynomial time and data, while there is no (polynomial-time or other) algorithm that constructs a characteristic sample for every language in the class, using a specific underlying representation of the languages in the class.

For the purpose of defining such a class, let $\varphi$ be a G{\"o}del numbering of all partial computable functions over the natural numbers, and let $\Phi$ be a corresponding Blum complexity measure. Here $\varphi_i$ refers to the $i$th partial computable function in the numbering~$\varphi$. Intuitively, $\Phi_i(j)$ is undefined if $\varphi_i(j)$ is undefined (i.e., the computation of $\varphi_i(j)$ does not terminate); otherwise $\Phi_i(j)$ is the number of computational steps required until the termination of the computation of $\varphi_i(j)$.
The set $K=\{k\mid\varphi_k(k)\mbox{ is defined}\}$ is called the halting set for $\varphi$; this set is recursively enumerable but membership in $K$ is not decidable.

We now define two languages for each natural number $k$:
\begin{eqnarray*}
L_{k,1}&=&a^*\cup\{b^k\}\\
L_{k,2}&=&\begin{cases}\{a^i\mid i\le\Phi_k(k)\}\cup \{b^k\}\,,&\mbox{if }k\in K\,,\\
a^*\cup\{b^k\}\,,&\mbox{otherwise}\,.
\end{cases}
\end{eqnarray*}

Note that $L_{k,1}=L_{k,2}$ if and only if $k\notin K$. Now let $\mathbb{L}$ consist of all languages $L_{k,q}$ for $k\in\mathbb{N}$ and $q\in\{1,2\}$.

There is an effective algorithm that decides membership in $L_{k,q}$, given $k\in\mathbb{N}$ and $q\in\{1,2\}$. To see that, note that, given $k$, $q$, and a word $w$, membership is trivial to decide when $q=1$ or when $w$ is not of the form $a^i$. If $w=a^i$, then $w$ belongs to $L_{k,2}$ if and only if the computation of $\varphi_k(k)$ does not terminate within fewer than $i$ steps, which can be checked effectively. 

Moreover, every language in $\mathbb{L}$ is regular and has a characteristic sample of size at most~2. In particular, $\{\langle b^k,1\rangle\}$ serves as a characteristic sample for $L_{k,1}$ (and thus also for $L_{k,2}$ in case $k\notin K$), while $\{\langle b^k,1\rangle,\langle a^{\Phi_k(k)+1},0\rangle\}$ is a characteristic sample for $L_{k,2}$ in case $k\in K$. Thus, using the above representation, the class $\mathbb{L}$ has polynomial-size (even constant-size) characteristic samples. However, there is no algorithm to construct such characteristic samples effectively, since otherwise such an algorithm could be used to decide membership in~$K$. (The latter can be verified by noting that $L_{k,2}\subseteq L_{k,1}$ for all $k$. Therefore, a system of characteristic samples would need to distinguish $L_{k,2}$ from $L_{k,1}$ (when $k\in K$) by either (i) a negative example of the form $\langle a^i,0\rangle$ for $L_{k,2}$, or (ii) a positive example of the form $\langle a^i,1\rangle$ for $L_{k,1}$, where $a^i\notin L_{k,2}$. Thus, the presence or absence of such example in the  characteristic samples for $L_{k,1}$, $L_{k,2}$ can be used to decide whether or not $k\in K$.)
\end{proof}

Since we are concerned with learning classes of automata, we now formulate the definition of \emph{efficient identification} directly over classes of automata.

\begin{defi}[efficient identification]\label{def:eff-ident}
    A class of automata $\class{M}$ over an alphabet $\Sigma$ 
     is said to be 
    \emph{efficiently identified} 
    if the following two requirements are met. 
    \begin{enumerate}
    \item
    There exists a polynomial time learning algorithm $\Infer:2^{(\Sigma^*\times\{0,1\})}\rightarrow \class{M}$ such that, for any sample $\mathcal{S}$, 
    we have $\mathcal{S}\subseteq\hatlang{\Infer(\mathcal{S})}$. 
    \item
    There exists a polynomial time  algorithm $\Char:\class{M}\rightarrow 2^{(\Sigma^*\times\{0,1\})}$ such that,
    for every $\aut{M}\in\class{M}$ and every sample $\mathcal{S}$ satisfying $\Char(\aut{M})\subseteq \aut{S} \subseteq \hatlang{\aut{M}}$, 
    the automaton $\Infer(\aut{S})$ recognizes the same language as $\aut{M}$. 
    \end{enumerate}
\end{defi}
When we apply this definition for a class of SFAs over a Boolean algebra $\alge{A}$ with domain~$\dom{D}$ and predicates $\class{P}$,
the characteristic sample is defined over the concrete set of letters $\dom{D}$ rather than the
set of predicates $\class{P}$ as this is the alphabet of the words accepted by an SFA. (Inferring an
SFA from a set of words labeled by predicates can be done using the methods for inferring DFAs, by considering
the alphabet to be the set of predicates.)

Throughout this section we study whether a class of SFAs $\class{M}$ 
is efficiently identifiable. That is, we are interested in the existence of algorithms
$\Infer_{\class{M}}$ and $\Char_{\class{M}}$ satisfying the requirements of Definition~\ref{def:eff-ident}.
In \autoref{subsec:necessary} we provide a necessary condition for a class of SFAs to be identified in the limit using polynomial time and data. In \autoref{subsec:sufficient} we provide a sufficient condition for a class of SFAs to be efficiently identifiable. On the positive side, we show  in \autoref{subsec:positive} that  the class of SFAs over the interval algebra is efficiently identifiable. On the negative side,  we show in \autoref{subsec:negative} that SFAs over the general propositional algebra cannot be identified in the limit using polynomial time and data. All classes of SFAs that we study are assumed to contain all basic SFAs (as per Definition~\ref{def:non-trivial}). 

\subsection{Efficient Identification of DFAs}
\label{sec:efficientidentificationDFAs}
Before investigating efficient identification of SFAs, it is worth noting that DFAs are efficiently identifiable. 
We state a result that provides more details about the nature of these algorithms, since we need it later, in \autoref{subsec:positive},
to provide our positive result. 
Intuitively, it says that 
there exists a valid learner such that 
if~$\aut{D}$ is a minimal DFA recognizing a certain language $L$ then 
the learner can infer $L$ from a characteristic sample consisting of
access words to each state of $\aut{D}$ and their extensions with distinguishing words (words showing each
pair of states cannot be merged) as well as one letter extensions of the access words that are required to retrieve the transition relation.
For completeness we give a proof of this theorem in Appendix~\ref{App:DFA}. 

\begin{thmC}[\cite{RPNI}]\label{prop:charDFA}
   % \itemI
    The class of DFAs is efficiently identifiable via
    procedures $\CharDFA$ and $\InferDFA$. 
    %\itemII
    Furthermore, these procedures satisfy that if 
    $\aut{D}$ is a minimal {and complete} DFA and
    $\CharDFA(\aut{D})=\mathcal{S}_\aut{D}$ then
     the following hold:
   \begin{enumerate}
          \item \label{item:lex}$\mathcal{S}_\aut{D}$ contains a prefix-closed set $A$ of access words. 
          Moreover,  $A$ can be chosen to contain only lex-access words,  i.e., only the lexicographically smallest access word for each state.
        \item For every $u_1,u_2\in A$ it holds that $u_1\not\sim_{\aut{S}_\aut{D}} u_2$. \label{item:accesswords}
        \item 
        For every $u,v\in A$ and $\sigma\in\Sigma$, if $\Delta(q_\iota, u\sigma)\neq \Delta(q_\iota, v)$ then $u\sigma\not\sim_{\aut{S}_\aut{D}}v$. \label{item:notsim_inS}
    \end{enumerate}
\end{thmC}

We briefly describe $\CharDFA$ and $\InferDFA$.

The algorithm $\CharDFA$ works as follows. It first creates a prefix-closed set of access words to states.
This can  be done by considering the graph of the automaton and running an algorithm for finding a spanning tree from the initial state. Choosing one of the letters on each edge, the access word for a state is obtained by concatenating the labels on the unique path of the obtained tree that reaches that state. If we wish to work with lex-access words, we can use a depth-first search algorithm that spans branches according to the order of letters in $\Sigma$, starting from the smallest. The labels  on the paths of the spanning tree constructed 
this way will form the set of lex-access words.

Let $S$ be the set of access words (or lex-access words).
Next the algorithm turns to find a distinguishing word $v_{i,j}$ for every pair of states $s_i,s_j\in S$ (where $s_i\neq s_j$).
It holds that any pair of states of the minimal DFA has a distinguishing word of size quadratic in the size of the DFA.
Let $E$ be the set of all such distinguishing words. 
 Then the algorithm returns the set $\mathcal{S}_\aut{D}=\{\langle w, \aut{D}(w) \rangle ~|~ w\in (S\cdot E) \cup (S \cdot \Sigma \cdot E)\}$ where $\aut{D}(w)$ is the label $\aut{D}$ gives~$w$ (i.e.,~$1$ if it is accepted, and $0$ otherwise). 
	It is easy to see that $\mathcal{S}_\aut{D}$ satisfies the properties of~Theorem~\ref{prop:charDFA}.

The algorithm $\InferDFA$, given a sample of words $\aut{S}$, infers 
from it in polynomial time a DFA that agrees with  $\aut{S}$. Moreover, if $\aut{S}$ subsumes
the characteristic set $\aut{S}_\aut{D}$ of a DFA $\aut{D}$ then  $\InferDFA$ returns a DFA that recognizes $\aut{D}$. 
Let $W$ be the set of words in the given sample $\aut{S}$ (without their labels). 
Let $R$ be the set of prefixes of $W$ and $C$ the set of suffixes of $W$.
Note that  $\epsilon\in R$ and $\epsilon \in C$.
Let $r_0,r_1,\ldots$ be some enumeration of $R$ and $c_0,c_1,\ldots$ some enumeration of $C$ where $r_0=c_0=\epsilon$.
The algorithm builds a matrix $M$ of size $|R|\times|C|$ whose entries take values in $\{0,1,?\}$, and sets the value of entry $(i,j)$ as follows. If $r_i c_j$ is not in $W$, it is set to $?$. 
Otherwise it is set to $1$ iff the word $r_i c_j$ is labeled $1$ in $\mathcal{S}$. 
We get that 
$r_i \sim_\mathcal{S} r_j$ iff
for every $k$ such  that both $M(i,k)$ and $M(j,k)$ are different than $?$ we have that $M(i,k)=M(j,k)$.
The algorithm 
sets $R_0=\{\epsilon\}$. Once $R_i$ is constructed, the algorithm tries to establish whether for $r\in R_i$ and $\sigma\in\Sigma$, $r\sigma$
is distinguished from all words in $R_i$. It does so by considering all other words $r'\in R_i$ and checking whether $r \sim_\mathcal{S} r'$. 
If $r\sigma$ is found to be distinct from all words in $R_i$, then $R_{i+1}$ is set to $R_i\cup\{r\sigma\}$. The algorithm proceeds until no new words are distinguished.\footnote{There is some resemblance between the  partial matrix and the apartness relation of~\cite{DBLP:conf/tacas/VaandragerGRW22}, used in query learning of DFAs.}
 
%Let $k$ be the iteration of convergence. 
Let $k$ be the minimal iteration such that $R_k=R_{k'}$ for all $k'>k$.
 If not all words in~$R_k$ are in $W$ (that is $M(i,0)=?$ for some $r_i\in R_k$), the algorithm returns the prefix-tree automaton. 
 Otherwise, the states of the constructed DFA are set to be the words in $R_k$. The initial state is $\epsilon$ and a state $r_i$ is classified as accepting 
iff $M(i,0)=1$ (recall that the entry $M(i,0)$ stands for the value of $r_i\cdot \epsilon$ in $\mathcal{S}$).
To determine the transitions, for every $r\in R_k$ and $\sigma \in \Sigma$, recall that there exists at least one state $r'\in R$ that cannot be distinguished from $r\sigma$.  
The algorithm then adds a transition from $r$ on $\sigma$ to $r'$.

\section{Necessary Condition}\label{subsec:necessary}

We make use of the following definitions. 
A sequence $\partition{\Gamma_1,\ldots, \Gamma_m}$ consisting of finite sets of concrete letters $\Gamma_i\subseteq {\dom{D}}$ is termed a \emph{concrete partition} of $\dom{D}$ if the sets are pairwise disjoint (namely $\Gamma_i\cap\Gamma_j= \emptyset$ for every $i\neq j$).
Note that we \underline{do not} require that in addition $\bigcup_{1\leq i\leq k}\Gamma_i=\dom{D}$. 
We use 
$\cpart{\dom{D}}$
%$\cpartitions{\dom{D}}{m}$ 
to define the set of all concrete partitions over~$\dom{D}$.
A~sequence of predicates $\partition{\psi_1,\ldots,\psi_k}$ 
over a Boolean algebra $\alge{A}$ on a domain $\dom{D}$ is termed a \emph{predicate partition} if $\sema{\psi_i}\cap\sema{\psi_j}=\emptyset$ for every $i\neq j$, and in addition $\bigcup_{1\leq i\leq k}\sema{\psi_i}=\dom{D}$. 
That is, 
here  we \underline{do} require the assignments to the predicates cover the domain.
We use 
$\ppart{\class{P}}$
%$\ppartitions{\class{P}}{k}$ 
to define the set of all predicate partitions over $\class{P}$.
For both concrete partition and predicate partition, we \underline{do not} require that the sets $\Gamma_i$ or $\sema{\psi_i}$ are non-empty.

\begin{defi}\label{def:gen-conc}
\ \ 
    \begin{itemize}
               \item \label{criterion:concertize}
            A function $f_c:\ppart{\class{P}}\rightarrow\cpart{\dom{D}}$ %from a predicate  partition  to a concrete partition 
            is termed \emph{concretizing} 
            if $f_c(\partition{\psi_1,\ldots,\psi_m})~=$ $\partition{\Gamma_1,\ldots, \Gamma_k }$ 
            implies $k=m$ and  $\Gamma_i\subseteq\sema{\psi_i}$ for all $1\leq i\leq m$.
       
        \item \label{criterion:generalize}
            A function $f_g:\cpart{\dom{D}}\rightarrow\ppart{\class{P}}$ %from a concrete partition  to a predicate partition  
            is termed \emph{generalizing} 
            if $f_g(\partition{\Gamma_1,\ldots, \Gamma_m })=\partition{\psi_1,\ldots,\psi_k}$ 
            implies $k=m$ and  $\sema{\psi_i} \supseteq \Gamma_i$ for all $1\leq i\leq m$.

    \end{itemize}
\end{defi}
Note that $f_g$ and $f_c$ are defined over partitions of any size. In Theorem~\ref{theorem:necessary-cond} we use their \emph{dyadic} restriction, that is, a concretizing and a generalizing functions that are defined only over partitions of size two. 

   We say that $f_g$ (resp. $f_c$) is \emph{efficient} if it can be computed in polynomial time. Note that if $f_c$ is efficient then the sets $\Gamma_i$ in the constructed concrete partition are of polynomial~size.

We are now ready to provide a necessary condition for identifiability in the limit using polynomial time and data.

\begin{thm} \label{theorem:necessary-cond}
   Let $\class{M}_\alge{A}$ be a class of SFAs over a Boolean algebra $\alge{A}$, that contain all basic SFAs over $\alge{A}$. If $\class{M}_\alge{A}$ is identified in the limit using polynomial time and data, then 
    there exist efficient dyadic concretizing and generalizing functions $\concretize_\alge{A} : \ppart{\class{P}}\rightarrow \cpart{\dom{D}}$ and $\generalize_\alge{A}:\cpart{\dom{D}}\rightarrow\ppart{\class{P}}$ 
    satisfying that 
    %\vspace{-2mm}
    \begin{center}
    %\begin{tabular}{l}
    if $\concretize_\alge{A} (\partition{\psi_1,\psi_2})=\partition{\Gamma_1,\Gamma_2}$ \\
    and $\generalize_\alge{A} (\partition{\Gamma'_1,\Gamma'_2})=\partition{\varphi_1,\varphi_2}$ \\
    where $\Gamma_i \subseteq \Gamma'_i$  for every $1\leq i\leq 2$ \\
    then $\sema{\varphi_i}=\sema{\psi_i}$ for every $1\leq i \leq 2$.
    %\end{tabular}
    \end{center}
    %\vspace{-4mm}    
\end{thm}

\begin{proof}
    Assume 
    that $\class{M}_\alge{A}$ is identified in the limit using polynomial time and data. 
    That is, there exist two algorithms
    $\pfcharsfa:\class{M}_\alge{A}\rightarrow 2^{\dom {D}^*\times\{0,1\}}$ and $\pfinfersfa:2^{\dom {D}^*\times\{0,1\}}\rightarrow\class{M}_\alge{A}$
    satisfying the requirements of Definition~\ref{def:ident-limit}.
    We show that efficient dyadic concretizing  and generalizing functions do exist.
    
    We start with the definition of $\concretize_\alge{A}$.
    Let $\partition{\varphi_1,\varphi_2}$ be the argument of $\concretize_\alge{A}$. Note that $\varphi_2=\neg\varphi_1$ by the definition of a predicate partition. The implementation of $\concretize_\alge{A}$ 
    invokes $\pfcharsfa$ on the
    SFA $\aut{M}_{\varphi_1}$ accepting all words of length one
    consisting of a concrete letter satisfying $\varphi_1$, as defined in~Definition~\ref{def:non-trivial}.
    Let $\aut{S}$ be the returned sample. Let $\Gamma_1$ 
    be the set of positively labeled words in the sample. Note that all such words are of size one, namely they are letters.
    Let $\Gamma_2$ be the set of letters that are  first letters in a  
    negative 
    word in the sample.
    Then $\concretize_\alge{A}$ returns $\langle \Gamma_1,\Gamma_2 \rangle$.

    We turn to the definition of $\generalize_\alge{A}$.
    Given $\partition{\Gamma_1,\Gamma_2}$ the implementation of $\generalize_\alge{A}$ invokes $\pfinfersfa$ on  sample
    $\aut{S} = \{ (\gamma, 1)~|~ \gamma \in \Gamma_1\}\cup \{(\gamma, 0)~|~\gamma\in \Gamma_2\}\cup\{(\gamma\gamma',0)~|~\gamma,\gamma'\in\Gamma_1\cup\Gamma_2\}$.
    That is,  all one-letter words in $\Gamma_1$ are positively labeled, all one-letter words in~$\Gamma_2$ are negatively labeled, and
    all words of length $2$ using some of the given concrete letters are negatively labeled.
    Let $\aut{M}$ be the returned SFA when given $\aut{S}'$, such that $\aut{S}'\supseteq \aut{S}$, as an input. Let $\Psi_1$ be the set
    of all predicates labeling some edge from the initial state to an accepting state, and let $\Psi_2$ be the set of all predicates labeling some edge from the initial state to a rejecting state. Let $\varphi=(\bigvee_{\psi\in\Psi_1} \psi) \wedge (\bigwedge_{\psi\in\Psi_2} \neg \psi)$.
    Then $\generalize_\alge{A}$ returns~$\langle \varphi,\neg\varphi\rangle$.

    It is not hard to verify that the constructed methods $\concretize_\alge{A}$ and $\generalize_\alge{A}$ satisfy the requirements of the theorem. 
    \end{proof}

The following example shows
the existence of functions $\concretize$ and $\generalize$ for the interval algebra.

\begin{exa}\label{ex:con_gen}
     Consider the class $\class{M}_{{\alge{A_\dom{N}}}}$ of SFAs over the algebra ${\alge{A_\dom{N}}}$ of Example~\ref{ex:SFA} 
     and consider the functions 
     $\concretize_{{{\alge{A_\dom{N}}}}}(\partition{~[d_1, d'_1),~[d_2, d'_2), \ldots,~[d_m, d'_m)}) = \partition{ \{ d_1 \},  \ldots, \{ d_m\}   }$
     and 
     $\generalize_{{{\alge{A_\dom{N}}}}}(\partition{ \Gamma_1 , \ldots, \Gamma_m}) = $ $\langle ~[\min{\Gamma_{1}}, \min{\hat{\Gamma_{1}}} ), ~[\min{\Gamma_{2}}, \min{\hat{\Gamma_{2}}} ),\ldots,  ~[\min{\Gamma_{m}}, \infty)\rangle$
     where  
     $\hat{\Gamma_i}=\bigcup_{j\neq i} \Gamma_j$ for every $1\leq i < m$. 
     Then, $\concretize_{{{\alge{A_\dom{N}}}}}$ and 
     $\generalize_{{{\alge{A_\dom{N}}}}}$ satisfy the variadic generalization of the conditions of Theorem~\ref{theorem:necessary-cond}.
\end{exa}

We would like to relate the necessary condition on the learnability of a class of SFAs over a Boolean algebra $\alge{A}$
to the learnability of the Boolean algebra $\alge{A}$ itself.
For this we need to first define efficient identifiability of a Boolean algebra $\alge{A}$. Since to learn an unknown predicate we need to supply two sets, one of negative examples and one of  positive examples, it makes sense to 
say that a Boolean algebra $\alge{A}$ with predicates $\class{P}$ over domain $\dom{D}$ is \emph{efficiently identifiable} if there exist
efficient dyadic concretizing and generalizing functions, $\concretize_\alge{A} : \ppart{\class{P}}\rightarrow \cpart{\dom{D}}$ and $\generalize_\alge{A}:\cpart{\dom{D}}\rightarrow\ppart{\class{P}}$ satisfying the criteria of Theorem~\ref{theorem:necessary-cond}. 
Using this terminology we can state the following corollary.

\begin{cor}
    Efficient identifiability of the Boolean algebra $\alge{A}$ is a necessary condition for identification 
    in the limit using polynomial time and data of  any class of SFAs over $\alge{A}$, that contains all basic SFAs over $\alge{A}$.
\end{cor}

%\sub
\section{Sufficient Condition}\label{subsec:sufficient}
We turn to discuss a sufficient condition for the efficient identifaibility of a class of SFAs $\class{M}_\alge{A}$ over a Boolean algebra $\alge{A}$.
To prove that $\class{M}_\alge{A}$ is efficiently identifiable, we need to supply two algorithms $\CharSFA_{\class{M}_\alge{A}}$
and $\InferSFA_{\class{M}_\alge{A}}$
as required in Definition~\ref{def:eff-ident}. The idea is to reduce
the problem to efficient identifiablity of DFAs, namely to use the algorithms $\CharDFA$ and $\InferDFA$
provided in Theorem~\ref{prop:charDFA}. The implementation of $\CharSFA$, given an SFA $\aut{M}$, will transform it
into a DFA $\aut{D}_\aut{M}$ by applying $\concretize_\alge{A}$ on the partitions induced by the states of the SFA.
The resulting DFA $\aut{D}_\aut{M}$ will not be equivalent to the given SFA $\aut{M}$, but it may be used to create a sample of words $\aut{S}_\aut{M}$
that is a characteristic set for~$\aut{M}$, see \autoref{fig:diagram}.

To implement $\InferSFA$ we would like to use $\InferDFA$ to obtain, as a first step, a DFA from
the given sample, then at the second step, apply $\generalize_\alge{A}$ on the concrete-partitions induced
by the DFA states. A subtle issue that we need to cope with is that inference should succeed
also on samples subsuming the characteristic sample. The fact that this holds for inference of 
the DFA does not suffice, since we are guaranteed that the inference of the DFA will not be confused
if the sample contains more labeled words, as long as the new words are over the same alphabet. In our
case the alphabet of the sample can be a strict subset
of the concrete letters $\dom{D}$ (and if $\dom{D}$ is infinite, this surely will be the case).
Example~\ref{ex:decontanimate} in \autoref{subsec:positive} illustrates this problem for the class of SFAs over a monotonic algebra $\alge{A_m}$, for which the respective methods $\concretize_\alge{A_m}$ and $\generalize_\alge{A_m}$ exist. 
So, we need an additional step to remove words from the given sample if they
are not over the alphabet of the characteristic sample. We call a method implementing this~$\Decontaminate_{\class{M}_\alge{A}}$.

\begin{figure}%[t]%[!htb]
	
	%\begin{minipage}{1.2\textwidth}
	\centering
	
	 \includegraphics[page=2, width=0.9\textwidth, clip, trim=117 180 100 500]{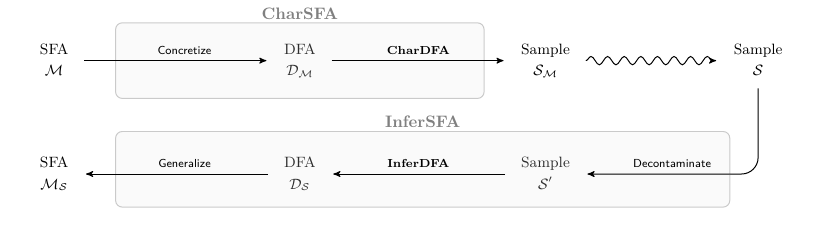}
	\caption{A schematic description of algorithms $\CharSFA$ and $\InferSFA$}\label{fig:diagram}
	%\end{minipage}%
\end{figure}

Formally, we first define the extension of 
$\concretize_\alge{A}$ and
$\generalize_\alge{A}$ to automata instead of partitions,
which we term $\Concretize_{\class{M}_\alge{A}}$ and $\Generalize_{\class{M}_\alge{A}}$ (with $\class{M}$ in the subscript).

\begin{figure}[t]%[!htb]
  
    \begin{minipage}{.5\textwidth}
            \centering
             \includegraphics[page=1, width=0.9\textwidth, clip, trim=90 580 280 85]{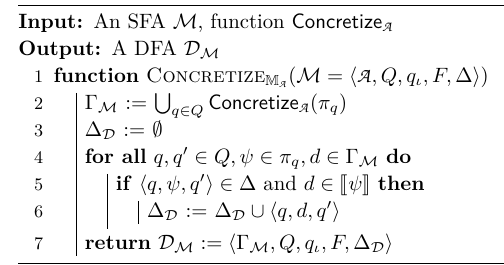}
    		  \captionof{algorithm}{$\Concretize_{\class{M}_\alge{A}}(\aut{M})$}\label{alg:conc}
    \end{minipage}%
    %\captionof{algorithm}{concretize}
    \hfill
    \begin{minipage}{.5\textwidth}
            \centering
    		  	 \includegraphics[page=1, width=0.95\textwidth, clip, trim=90 420 260 245]{algorithms.pdf}
    	\captionof{algorithm}{$\Generalize_{\class{M}_\alge{A}}(\aut{M})$}\label{alg:gen}
    \end{minipage}
    %\captionof{algorithm}{generalize}
\end{figure}

\begin{itemize}
    \item 
    The formal definition of $\Concretize_{\class{M}_\alge{A}}$ is given in Algorithm~\ref{alg:conc}.
    Let $\aut{M}=(\alge{A},Q,q_\iota ,F,\Delta)$ be an SFA. 
    Then $\Concretize_{\class{M}_\alge{A}}(\aut{M})$ is the DFA $\aut{D }_\aut{M}   = (\Sigma, Q,q_\iota ,F,\Delta_\aut{D})$ where $\Delta_\aut{D}$ is defined as follows. 
    For each state $q\in Q$
    let $\pi_q=\langle \psi_1,\ldots,\psi_m\rangle$ be the predicate partition consisting of all predicates labeling a transition exiting $q$ in $\aut{M}$. Intuitively, in $\aut{D}$, the outgoing transitions of each state $q$
    correspond to $\concretize_\alge{A}(\pi_q)$. That is, 
    let $\concretize_\alge{A}(\pi_q) = \langle \Gamma_1, \ldots, \Gamma_m \rangle$. Then,
    if $\langle q, \psi_i, q'  \rangle \in \Delta$, then $\langle q, \gamma, q'  \rangle \in \Delta_\aut{D}$ for every $\gamma\in \Gamma_i$.

    \item  
    The formal definition of $\Generalize_{\class{M}_\alge{A}}$ is given in Algorithm~\ref{alg:gen}.
    Let $\aut{D }   = (\Sigma, Q,q_\iota ,F,\Delta_\aut{D})
    $ be a DFA. We define $\Generalize_{\class{M}_\alge{A}}(\aut{D})$ with respect to an algebra $\alge{A}$ as follows. Let $\aut{M}=(\alge{A},Q,q_\iota ,F,\Delta_\aut{M})$ where 
    $\Delta_\aut{M}$ is defined as follows. For each state $q\in Q$ let $\langle \Gamma_1, \ldots, \Gamma_m  \rangle$ be the concrete partition consisting of letters labeling outgoing transitions from $q$. 
    Note that  $\partition{\Gamma_1, \ldots, \Gamma_m}$ is  a concrete partition, since $\aut{D}$ is a DFA.
    Let $\generalize_\alge{A}(\langle \Gamma_1, \ldots, \Gamma_m \rangle) = \langle \psi_1, \ldots, \psi_m \rangle$. Then, $\langle q, \psi_i, q'  \rangle \in \Delta_\aut{M}$ 
    if $\Gamma_i$ is the set of letters labeling transitions from $q$ to $q'
    $ in $\aut{D}$.  
\end{itemize}

We are now ready to define the conditions the \emph{decontaminating} function has to satisfy. We recall that the role of the decontaminating function is to identify words in the sample that are not over the alphabet $\Gamma_\aut{D}$ of the characteristic sample  (note that $\Gamma_\aut{D}$ is not known to the decontaminating function).

\begin{defi}\label{def:decont} 
        A function $f_d:2^{(\dom{D}^*\times\{0,1\})}\rightarrow 2^{(\dom{D}^*\times\{0,1\})}$  is called \emph{decontaminating}
    for a class of SFAs $\class{M}$ and a respective $\Concretize_\class{M}$ function
    if the following holds.  
    Let $\aut{M}\in\class{M}$ be an SFA, and let $\aut{D} = \Concretize_\class{M}(\aut{M})$. Let $\aut{S_\aut{D}} = \CharDFA(\aut{D})$.
    Then, for every $\aut{S}' \supseteq \aut{S}_\aut{D}$ such that $\aut{S}'$ agrees with $\aut{M}$, it holds that 
    $\aut{S}_\aut{D} \subseteq f_d (\aut{S}') \subseteq (\aut{S}' \cap \Gamma_\aut{D})$, where $\Gamma_\aut{D}$ is the alphabet~of~$\aut{S}_\aut{D}$.
\end{defi}

As before we say that $f_d$ is \emph{efficient} if it can be computed in polynomial time. 

\begin{exa}\label{ex:decon1}
Intuitively, $\InferDFA$ is only promised to be correct if it is applied on a sample set $\aut{S}'$ over the alphabet of the original DFA. For DFAs, this is always the case. However for SFAs, the concrete alphabet $\mathbb{D}$ usually contains more letters than appear in the characteristic set $\aut{S}$. 
If $\Gamma_D$ is the set of letters in $\aut{S}$ then  $\aut{S}'$ might contain letters in $\mathbb{D}\setminus\Gamma_D$, i.e., letters that are not from the alphabet of the characteristic set. 
Consider, for example, the characteristic set $\aut{S}$ over the interval algebra: $$\aut{S} = \{ \langle \epsilon, 0  \rangle, \langle 0, 1  \rangle, \langle 100, 0  \rangle, \langle 200, 0  \rangle,
     \langle 0\cdot 0, 1  \rangle, \langle 0\cdot 100, 1  \rangle,\langle 0\cdot 200, 0  \rangle
    \}$$ 
    and consider the set $\aut{S}' \supseteq \aut{S}$ that contains, in addition to the words in $\aut{S}$, also the word $150\cdot 100$. Since the letter $150$ is not part of any word in the original set $\aut{S}$, we cannot apply $\InferDFA$ as it is, but we first need to remove it from $\aut{S}'$. Note that words that are not in~$\aut{S}$ but are over the alphabet $\{ 0,100,200\}$ do not pose a problem, as $\InferDFA$ can handle supersets over the same alphabet as the set $\aut{S}$. 
See Examples~\ref{ex:infer_interval} and~\ref{ex:decontanimate} in \autoref{subsec:positive} for more details. 

\end{exa}

We now provide the sufficient condition for efficient identifiability. 

\begin{thm} \label{theorem:sufficient}
      Let $\class{M}_\alge{A}$ be a class of SFAs over a Boolean algebra $\alge{A}$. If there exist an efficient decontaminating function $\Decontaminate_{\class{M}_\alge{A}}$ and 
      efficient functions $\concretize_\alge{A}$ and $\generalize_\alge{A}$ satisfying that 
     % \vspace{-2mm}
     \begin{center}
      %   \begin{tabular}{l}
        if $\concretize_\alge{A}(\partition{\psi_1,\ldots,\psi_m})=\partition{\Gamma_1,\ldots,\Gamma_m}$ \\
        and $\generalize_\alge{A}(\partition{\Gamma'_1,\ldots,\Gamma'_m})=\partition{\varphi_1,\ldots,\varphi_m}$\\
        where  $\Gamma_i \subseteq \Gamma'_i$  for every $1\leq i\leq m$\\
        then  $\sema{\varphi_i}=\sema{\psi_i}$ for every $1\leq i \leq m$
       % \end{tabular}
    \end{center}
        %  \vspace{-2mm}
        then the class $\class{M}_\alge{A}$
     is efficiently identifiable.
\end{thm}

Given functions $\concretize_{\alge{A}}$, $\generalize_{\alge{A}}$ and
$\Decontaminate_{\class{M}_\alge{A}}$
 for a class~$\class{M}_\alge{A}$ of SFAs over a Boolean algebra $\alge{A}$, meeting the criteria of Theorem~\ref{theorem:sufficient}, 
we show that $\class{M}_\alge{A}$  can be efficiently identified by providing two algorithms $\CharSFA$ and $\InferSFA$,  described below. 
These algorithms make use of the respective algorithms  $\CharDFA$ and $\InferDFA$ guaranteed in Theorem~\ref{prop:charDFA}, as well as the methods 
provided by the theorem.

We briefly describe these two algorithms, and then turn to prove Theorem~\ref{theorem:sufficient}.
The algorithm $\CharSFA$ %(Alg.\rfalgCharacteristic), 
receives an SFA ${\aut{M}\in \class{M}}$,
and returns a characteristic sample for it. 
It does so by 
applying $\Concretize_{\class{M}_\alge{A}}(\aut{M})$ (Algorithm~\ref{alg:conc})  to
construct a DFA $\aut{D}_\aut{M}$ 
and generating the sample $\mathcal{S}_\aut{M}$ using the algorithm $\CharDFA$ applied on the DFA $\aut{D}_\aut{M}$. %, and returns it.

Algorithm $\InferSFA$, given a sample set~$\mathcal{S}$, if~$\aut{S}$ subsumes a characteristic set of an
SFA~$\aut{M}$, returns an equivalent SFA. Otherwise $\InferSFA$ returns an SFA that agrees with the sample~$\aut{S}$.
First, it 
applies $\Decontaminate_{\class{M}_\alge{A}}$ to
find a subset $\aut{S'}\subseteq\aut{S}$ over the alphabet of the subsumed characteristic sample, if such a subsumed sample exists.
Then it uses $\aut{S'}$ to construct a DFA by applying the inference algorithm $\InferDFA$ on $\aut{S'}$. 
From this DFA it constructs an SFA, $\aut{M}_\mathcal{S}$, by applying  $\Generalize_{\class{M}_\alge{A}}$  (Algorithm~\ref{alg:gen}). 
If the resulting automaton disagrees with the given sample it resorts to returning the prefix-tree automaton. 
In order to construct the symbolic prefix-tree automaton we 
first construct the prefix-tree DFA $\aut{A}$ for the set $\mathcal{S}$, and then apply $\Generalize_{\class{M}_\alge{A}}(\aut{A})$ to get an SFA that agrees with $\mathcal{S}$.

In brief, we define: 
\begin{itemize}
\item
$\CharSFA (\aut{M}) = \CharDFA (\Concretize_{\class{M}_\alge{A}} (\aut{M}) )$
\item
$\InferSFA(\aut{S}) \!=\! \!
    \begin{cases}
      \aut{M}_\aut{S} \!:= \Generalize_{\class{M}_\alge{A}} (\InferDFA(\Decontaminate_{\class{M}_\alge{A}}(\mathcal{S})))
 & \text{if }\mathcal{S} \!\subseteq\hatlang{\aut{M}_\mathcal{S}}\\
      \text{The symbolic prefix-tree automaton of $\aut{S}$} & \text{otherwise}
    \end{cases}$
\end{itemize}

In \autoref{subsec:positive} we provide
methods $\concretize_{\alge{A}}$, $\generalize_{\alge{A}}$ and
$\Decontaminate_{\class{M}_\alge{A}}$ for SFAs over monotonic algebras,
deriving their identification in the limit
result.
We now prove~Theorem~\ref{theorem:sufficient}.

\begin{proof}[Proof of Theorem~\ref{theorem:sufficient}] 
Given functions $\concretize_\alge{A}$, $\generalize_\alge{A}$, and $\Decontaminate_{\class{M}_\alge{A}}$,  
we show that the algorithms 
$\CharSFA$ and $\InferSFA$ satisfy the requirements of Definition~\ref{def:eff-ident}.

For the first condition, given that  $\CharDFA$, $\Decontaminate_{\class{M}_\alge{A}}$ and $\Generalize_\alge{A}$ run in polynomial time, and that the prefix-tree automaton can be constructed in polynomial time, 
it is clear that so does $\InferSFA$.
In addition, the test performed in the definition of $\InferSFA$
ensures the output agrees with the sample. 

For the second condition, note that the sample generated by $\CharSFA$ is polynomial in the size of $\aut{D}_\aut{M}$, from the correctness of $\CharDFA$. In addition, since $\Concretize_\alge{A}$ is efficient, $\aut{D}_\aut{M}$ is polynomial in the size of $\aut{M}$, and thus $\aut{S}_\aut{M}$ generated by $\CharSFA$ is polynomial in~$\aut{M}$ as well. 
It is left to show that given $\mathcal{S}_\aut{M}$ is the concrete sample produced by $\CharSFA$ when running on an SFA~$\aut{M}$, then when $\InferSFA$ runs on any sample~$\mathcal{S} \supseteq \mathcal{S}_\aut{M}$ it returns an SFA for~$\lang{\aut{M}}$.
Since $\Decontaminate_{\class{M}_\alge{A}}$ is a decontaminating function, and $\aut{S}\supseteq \aut{S}_\aut{M}$,
it holds that the set $\aut{S'}=\Decontaminate_{\class{M}_\alge{A}}(\aut{S})$ is such that $\aut{S'}\supseteq \aut{S}_\aut{M}$ and is only over the alphabet $\Gamma_\aut{M}$, which is  the alphabet of the DFA $\aut{D}_\aut{M}$ generated in Algorithm~\ref{alg:conc}.

From the correctness of $\InferDFA$,
given $\aut{S'}\supseteq\aut{S}_\aut{M}$, 
applying $\InferDFA$ on the output~$\aut{S}'$ of $\Decontaminate_{\class{M}_\alge{A}}$
results in a DFA~$\aut{D}$ 
that is equivalent to~$\aut{D}_\aut{M}$ constructed in Algorithm~\ref{alg:conc}. 
Since~$\aut{D}_\aut{M}$ is complete with respect to its alphabet $\Gamma_\aut{M}$, for state~$q$ of~$\aut{D}$, the concrete partition $\langle \Gamma_1,\ldots,\Gamma_n\rangle$  generated in Algorithm~\ref{alg:gen}, line~\rflinegenGammai,
covers $\Gamma_\aut{M}$ and subsumes the output of $\concretize_{\class{M}_\alge{A}}$ on~$\pi_q$ (Algorithm~\ref{alg:conc}, line~\rflinesigmaConc).
 Thus, since $\generalize_\alge{A}$
and $\concretize_\alge{A}$ satisfy the criteria of Theorem~\ref{theorem:sufficient}, it holds that the constructed predicates agree with the original predicates. 
In addition, since~$\aut{S}$, and therefore $\aut{S'}$, agrees with~$\aut{M}$, the test performed in the definition of $\InferSFA$ 
succeeds and the returned SFA is 
equivalent to $\aut{M}$.
\end{proof}

%\sub
\section{Positive Result}\label{subsec:positive}
We present the following positive result regarding monotonic algebras.

\begin{thm}\label{prop:sample_monotonic}%
    Let $\class{M}_{\alge{A_m}}$ be the set of SFAs over a monotonic Boolean algebra $\alge{A_m}$. 
    Then $\class{M}_{\alge{A_m}}$ is efficiently identifiable.
\end{thm}

In order to prove Theorem~\ref{prop:sample_monotonic}, we show that the sufficient condition holds for the case of monotonic algebras. 
Example~\ref{ex:infer_interval} demonstrates how to apply $\CharSFA$ and $\InferSFA$ in order to learn an SFA over the algebra~$\alge{A_\dom{N}}$.

\begin{prop} \label{lemma:necessary_monotonic}
There exist functions
$\concretize_{\alge{A_m}}$ and 
 $\generalize_{\alge{A_m}}$ for 
 a monotonic Boolean algebra $\alge{A_m}$, satisfying the criteria of Theorem~\ref{theorem:sufficient}. 
\end{prop}

\begin{proof}
Let $\dom{D}$ be the domain of $\alge{A_m}$. 
We provide the functions $\concretize_{\alge{A_m}}$ and 
 $\generalize_{\alge{A_m}}$
 and prove that the criteria of Theorem~\ref{theorem:sufficient} hold for them.
For ease of presentation, for the function $\concretize_{\alge{A_m}}$  we consider 
basic predicates. Note that for monotonic algebras, basic predicates are in fact intervals, as a conjunction of intervals is an interval. We can assume all predicates are basic since, as we show in Lemma~\ref{lemma:dnf_monotonic}, 
 for monotonic algebras the transformation from a general formula to a DNF formula of basic predicates is linear. Then, each basic predicate in the formula corresponds to a different predicate in the predicate partition. 
The definitions 
of $\concretize_{\alge{A_m}}$ and 
 $\generalize_{\alge{A_m}}$
are generalizations of the functions $\concretize_{\alge{A_\dom{N}}}$ and  $\generalize_{\alge{A_\dom{N}}}$ given in Example~\ref{ex:con_gen}.
We define 
$\concretize_{\alge{A_m}}(\partition{ \psi_1 ,\ldots \psi_m}) = \partition{\Gamma_1, \ldots , \Gamma_m}$ where we set $\Gamma_i = \{\ \min \{ {d\in\dom{D}}\,|\,{d\in\sema{\psi_i}}  \}\ \} $ for $1\leq i \leq m$.
 Since $\alge{A_m}$ is monotonic, $\Gamma_i$ is well-defined and contains a single element, 
 thus $\concretize_{\alge{A_m}}$ is an efficient concretizing function.

We define $\generalize_{\alge{A_m}}(\partition{\Gamma_1,\ldots, \Gamma_m}) = \partition{\psi_1, \ldots, \psi_m}$, where $\psi_i$ is defined as follows. 
Let $\Gamma= \bigcup_{1\leq i\leq m} \Gamma_i$. 
First, for all $1\leq i\leq m$ we set $\psi_i = \bot$. 
Then, we iteratively look for the minimal element $\gamma\in\Gamma$. Let~$i$ be such that $\gamma\in \Gamma_{i}$, and let~$\gamma'$ be the minimal element in $\Gamma$ satisfying $\gamma'\notin \Gamma_{i}$. We then set  $\psi_{i} = \psi_{i} \vee [\gamma, \gamma')$, and remove all elements $\gamma \leq \gamma'' < \gamma'$ from~$\Gamma$. We repeat the process until for the found $\gamma\in\Gamma_{j}$, there is no $\gamma' > \gamma$ such that $\gamma'\notin\Gamma_{j}$. 
In that case, we define  $\psi_{j} = \psi_{j} \vee [\gamma, \dmax)$.
Then, $\Gamma_i \subseteq \sema{\psi_i} $ and the predicates are disjoint,  
thus $\generalize_{\alge{A_m}}$ is an efficient generalizing function. See Example~\ref{exa:genIntervals} and \autoref{fig:genIntervals}.

Now, let $\partition{\Gamma_1, \ldots, \Gamma_m}$ be the concrete partition obtained from $\concretize_{\alge{A_m}}$ when applied on the predicate partition $\partition{\psi_1,\ldots,\psi_m}$. Assume further that the predicate partition $\partition{\Gamma'_1, \ldots, \Gamma'_m}$ satisfies  $\Gamma_i \subseteq \Gamma'_i \subseteq \sema{\psi_i}$ for $1\leq i \leq m$. In particular, 
$\min(\Gamma'_i) = \min(\Gamma_i)$, since~$\Gamma_i$ contains the minimal  
elements in $\sema{\psi_i}$, and $\Gamma_i \subseteq \Gamma'_i \subseteq \sema{\psi_i}$. Thus applying $\generalize_{\alge{A_m}}$ will result in the same interval, satisfying the criterion 
of Theorem~\ref{theorem:sufficient}.
\end{proof}

\begin{figure}
    \centering
    \subfloat[]{\label{figexa-1}\includegraphics[scale=.35]{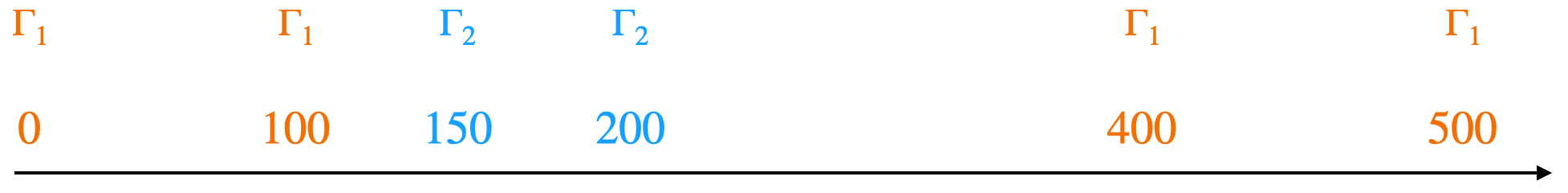}}\\
      \centering 
      \subfloat[]{\label{figexa-3}\includegraphics[scale=.35]{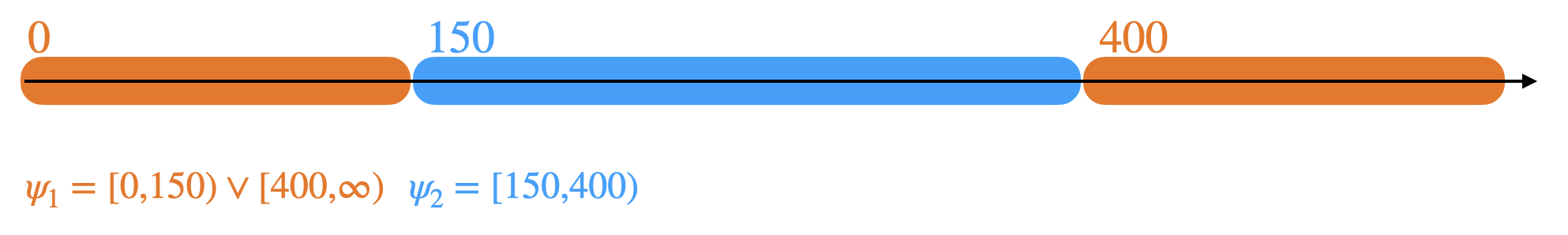}}
    \caption{Visualisation of Example~\ref{exa:genIntervals}}
    \label{fig:genIntervals}
\end{figure}

\begin{exa}\label{exa:genIntervals}
     Let $\Gamma_1 = \{ 0, 100, 400, 500 \}$ and $\Gamma_2 = \{ 150, 200 \}$ over the algebra ${{\alge{A_\dom{N}}}}$ with domain $\dom{N}\cup\{ \infty \}$, see Figure~\ref{figexa-1}. Then, $\generalize_\alge{A_\dom{N}}$ sets $\Gamma = \{0, 100, 150, 200, 400, 500 \}$, and finds the minimal element in $\Gamma$, which is $0$. Since $0\in\Gamma_1$, it then looks for the minimal element $\gamma\in\Gamma$ such that $\gamma\notin\Gamma_1$, and finds $150\in\Gamma_2$. Therefore $\psi_1 = [0, 150)$ and we remove from $\Gamma$ all elements in $[0, 150)$, that is, we remove $0$ and $100$, and we have $\Gamma = \{ 150, 200, 400, 500 \}$. Next, it finds the minimal element in the updated $\Gamma$, which is $150$ and is in $\Gamma_2$. The minimal element in $\Gamma\setminus \Gamma_2$ is $400$. 
     Then, $\psi_2$ is set to be $\psi_2 = [150, 400)$ and $\Gamma = \{400, 500\}$.
     Now, $\psi_1 = [0, 150) \vee [400, \infty)$ since $400\in \Gamma_1$ and there is no greater element that is not in $\Gamma_1$ (Figure~\ref{figexa-3}).
\end{exa}

To show that any class of SFAs $\class{M}_\alge{A_m}$ over a monotonic algebra $\alge{A_m}$
is efficiently identifiable, we define in Algorithm~\ref{alg:dec} an algorithm that implements a decontaminating function
$\Decontaminate_{\class{M}_\alge{A_m}}$, 
that fulfills the requirements of Theorem~\ref{theorem:sufficient}.
Loosely speaking, the idea of the algorithm is to simultaneously collect elements into two sets $A_w$ and $\Sigma'$
such that~$A_w$ will consist of the minimal representative 
 according to the lexicographic order of each
equivalence class in $\sim_\aut{S}$ and $\Sigma'$
will consist of minimal letters aiding to distinguish these words. 
When this process terminates the algorithm returns the subset of words in the sample
that consist of only letters in $\Sigma'$.

\begin{figure}[t]%[!htb]
	
	%\begin{minipage}{1.2\textwidth}
	\centering
	
	 \includegraphics[page=2, width=0.9\textwidth, clip, trim=110 435 120 155]{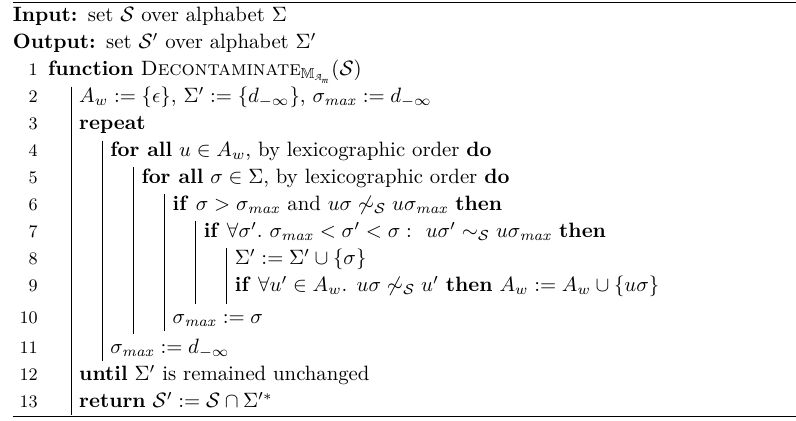}
	\captionof{algorithm}{$\Decontaminate_{\class{M}_{\alge{A_m}}}$ -- finding the necessary letters for a characteristic sample}\label{alg:dec}
	%\end{minipage}%
\end{figure}

\begin{lem}\label{lemma:induction}
Assume the input to $\Decontaminate_{\class{M}_{\alge{A_m}}}$ is $\aut{S}$, where $\aut{S} \supseteq \aut{S}_\aut{M}$
for some $\aut{M}\in\class{M}_{\alge{A_m}}$ such that $\aut{S}_\aut{M}= \CharDFA(\Concretize_{\class{M}_{\alge{A_m}}} (\aut{M}) )$, and
$\aut{D}_\aut{M}=\Concretize_{\class{M}_{\alge{A_m}}}(\aut{M})$ is over the alphabet $\Gamma_\aut{M}$.
Then for $\Sigma'$ constructed by $\Decontaminate_{\class{M}_{\alge{A_m}}}$ (Algorithm~\ref{alg:dec}) it holds that $\Sigma'\!=\!\Gamma_{\aut{M}}$.
\end{lem}

\begin{proof} Let $\aut{M}=(\alge{A},Q,q_\iota ,F,\Delta_\aut{M})$, $\aut{D}_\aut{M}=\Concretize_{\class{M}_{\alge{A_m}}} (\aut{M})$ 
and $\aut{S}_\aut{M}= \CharDFA(\aut{D}_\aut{M})$.
Assume $\aut{D}_\aut{M}=(\Gamma_\aut{M},Q,q_\iota ,F,\Delta_\aut{D})$.
We inductively show that for $\Decontaminate_{\class{M}_{\alge{A_m}}}$
given in Algorithm~\ref{alg:dec}, if its input $\aut{S}$ satisfies $\aut{S}\supseteq \aut{S}_\aut{M}$ 
then the set $A_w$ is exactly the set of all lex-access words (lexicographically smallest access words) 
of states in $\aut{D}_\aut{M}$ and that $\Sigma' = \Gamma_{\aut{M}}$
(where $\Gamma_{\aut{M}}$ is the alphabet of $\aut{D}_\aut{M}$).

First, we show that every $u\in A_w$ is a lex-access word and that $\Sigma' \subseteq \Gamma_{\aut{M}}$.

For the base case, we consider $A_w = \{ \epsilon \}$ and $\Sigma' = \{\dmin  \}$. 
From item~\ref{item:lex} of Theorem~\ref{prop:charDFA}, we can assume access words are minimal according to the lexicographic order. Thus, $\epsilon\in A_w$ is indeed a lex-access word (of the state $q_\iota$). 
For $\dmin\in\Sigma'$, it holds that
$\Gamma_{\aut{M}}$ contains the minimal element of $\dom{D}$ since it contains all concretizations  of intervals, the SFA is complete and $\concretize_{\alge{A_m}}$ returns the minimal element of each interval. 
Therefore $\dmin\in\Gamma_{\aut{M}}$.

For the induction step, 
 assume that $A_w$ contains only lex-access words and that the current $\Sigma'$ is a subset of $\Gamma_{\aut{M}}$.
Then, when considering $u\in A_w$ in line~4, %\ref{item:iter}, 
 it holds that $u$ is a lex-access word of some state $q$.
Then, $\sigma$ is added to $\Sigma'$ only if  $u\sigma\not\sim_\aut{S}u\dmin$. Since $\aut{S}$ agrees with $\aut{M}$, it holds that $\Delta_\aut{M}(q_\iota, u\sigma)\neq \Delta_\aut{M}(q_\iota, u\dmin) $ 
 and $\sigma$ is a minimal element with that property. Then, $\sigma$ must be a minimal element of an interval labeling an outgoing transition from $q$, therefore is in $\Gamma_{\aut{M}}$. Inductively this holds for all elements added to $\Sigma'$ in the current iteration.
 This proves that $\Sigma' \subseteq \Gamma_{\aut{M}}$.
 Assume now that $A_w$ contains only lex-access words and let $u\sigma$ be a word added to $A_w$ in line~9. 
 Thus, for all $u'\in A_w$ it holds that $u\sigma\not\sim_{\aut{S}} u'$.
 
 \begin{itemize}
 \item []
 \textbf{Claim.} In this setting, $u\sigma\not\sim_{\aut{S}} u'$ implies $u\sigma\not\sim_{\aut{S}_\aut{M}} u'$.
 \item []
 \textbf{Proof.}
 Since we assume all words already in $A_w$ are lex-access words,
 then in particular $u$ is a lex-access word. In addition, $\sigma\in\Sigma'$ and thus $\sigma\in\Gamma_{\aut{M}}$.
 Since $u\sigma\not\sim_{\aut{S}} u'$
 there is $z\in\Sigma^*$ such that $\langle u\sigma z, b_1 \rangle, \langle u'z, b_2 \rangle \in\aut{S}$ and $b_1\neq b_2$. 
 Since $\aut{S}$ agrees with $\aut{M}$ it holds that $\Delta_\aut{M}(q_\iota, u\sigma) \neq \Delta_\aut{M}(q_\iota, u')$.
 Now, $u\sigma$ and $u'$ are all in $\Gamma_{\aut{M}}$ since from line~9 %\ref{item:access_lex} 
 we have $A_w\subseteq\Sigma'^*$, and thus
 $\Delta_ {\aut{D}_\aut{M}}(q_\iota, u\sigma) \neq \Delta_{\aut{D}_\aut{M}}(q_\iota, v)$, and from item~\ref{item:notsim_inS} of Theorem~\ref{prop:charDFA}
it holds that $u\sigma\not\sim_{\aut{S}_\aut{M}}u'$. This proves the claim.
\end{itemize}

Then, for all $u'\in A_w$ we have $\Delta_{\aut{D}}(q_\iota, u\sigma) \neq  \Delta_{\aut{D}}(q_\iota, u')$. 
Then, since we traverse words and letters in lexicographic order, 
$u\sigma$ is a lex-access word for $\Delta_{\aut{D}}(q_\iota, u\sigma)$. 
This concludes the first direction.

For the second direction, we show that every lex-access word is in $A_w$ and that $\Gamma_{\aut{M}}\subseteq \Sigma'$.

Let $u\sigma$ be a lex-access word (for $\epsilon$ it holds that $\epsilon\in A_w$). For all lex-access words $u'$ found in previous iterations it holds that $u\sigma\not\sim_{\aut{S}_\aut{M}} u'$ from item~\ref{item:accesswords} of Theorem~\ref{prop:charDFA}, and thus $u\sigma\not\sim_{\aut{S}} u'$ since $\aut{S}_\aut{M}\subseteq\aut{S}$. Then, $u\sigma$ satisfies the condition of line~9 %\ref{item:access_lex} 
and is added to $A_w$. 

To prove $\Gamma_{\aut{M}}\subseteq\Sigma'$, 
let $\sigma\in\Gamma_{\aut{M}}$. From the construction of $\concretize_{\alge{A_m}}$ it holds that~$\sigma$ is the left endpoint of some interval that is an outgoing transition from $q_\iota$. Then, indeed~$\sigma$ is found in the first iteration of line~4. %\ref{item:iter}. 
 Inductively, let~$\sigma$ label an outgoing transition of~$q$ for some~$q\in Q$ and let~$u_q$ be the lex-access word of~$q$. 
 Since~$A_w$ contains all lex-access words, it holds that $u_q\in A_w$, and then the outgoing transitions of~$q$ will be considered in some following iteration. Thus, all minimal letters indicating new intervals are added to~$\Sigma'$ and we have that $\Gamma_{\aut{M}}\subseteq\Sigma'$. We conclude that $\Sigma' = \Gamma_{\aut{M}}$. 
\end{proof}

\begin{prop}\label{lemma:sufficient_monotonic}
The sufficient condition of Theorem~\ref{theorem:sufficient} holds for the class $\class{M}_{\alge{A_m}}$ of SFAs over a monotonic Boolean algebra $\alge{A_m}$. 
\end{prop}

\begin{proof}
In Proposition~\ref{lemma:necessary_monotonic} we have shown that there exist functions
$\concretize_{\alge{A_m}}$ and 
 $\generalize_{\alge{A_m}}$ for 
 a monotonic Boolean algebra $\alge{A_m}$, satisfying the criteria of Theorem~\ref{theorem:sufficient}.
 It is left to show that $\Decontaminate_{\class{M}_{\alge{A_m}}}$ is an efficient decontaminating function. 
Assume that $\aut{S}\supseteq \aut{S}_\aut{M}$ where
$ \aut{S}_\aut{M}= \CharDFA(\Concretize_{\class{M}_{\alge{A_m}}}\!(\aut{M}) )$, 
and $\Concretize_{\class{M}_{\alge{A_m}}}\!(\aut{M})$  is over alphabet $\Gamma_{\aut{M}}$.
In Lemma~\ref{lemma:induction} we showed that under these assumptions it holds that the alphabet $\Sigma'$ of the returned sample $\aut{S'}$  is $\Gamma_{\aut{M}}$. Then, for the set $\aut{S'}$ returned in line~13 %\ref{item:returnS'} 
(Algorithm~\ref{alg:dec}) it holds that  $\aut{S'} = \aut{S}\cap \Gamma_{\aut{M}}^* $. Since $\aut{S} \supseteq \aut{S}_\aut{M}$ and $\Gamma_{\aut{M}}^* \supseteq \aut{S}_\aut{M}$, it holds that $\aut{S'} \supseteq \aut{S}_\aut{M}$ and $\aut{S'}$ is defined over the alphabet $\Gamma_{\aut{M}}$. Therefore, $\Decontaminate_{\class{M}_{\alge{A_m}}}$
 is a decontaminating~function.~In addition, 
it runs in time  polynomial in the size of $\aut{S}$, thus the conditions of Theorem~\ref{theorem:sufficient} are~met. 
\end{proof}

We now provide an example for efficient identification of SFAs over the interval algebra (Example~\ref{ex:infer_interval}), and an example for the need of a decontaminating function (Example~\ref{ex:decontanimate}). 

\begin{exa} \label{ex:infer_interval}\label{example:mono}
     Continuing Example~\ref{ex:con_gen},  %and~\ref{ex:sample}
     let $\aut{M}$ be the SFA from Figure~\ref{fig:SFA} and consider the class $\class{M}_{{\alge{A_\dom{N}}}}$ of SFAs over the interval algebra. Algorithm $\CharSFA$
     computes the set $\Gamma_{\aut{M}}$ using the function $\concretize_{\class{M}_{{\alge{A_\dom{N}}}}}$ given in Example~\ref{ex:con_gen}.
     That is, for the predicates labeling outgoing transitions from $q_0$
    we have $\concretize_{{\alge{A_\dom{N}}}}(\partition{[0,100), [100, \infty)}  ) = \partition{\{ 0 \}, \{ 100 \}}$; and for outgoing transitions from $q_1$, it holds that $\concretize_{{\alge{A_\dom{N}}}}(\partition{[0,200), [200, \infty)}  ) = \partition{\{ 0 \}, \{ 200 \}}$. Then,  $\Gamma_{\aut{M}}=\{0, 100, 200  \}$, and $\CharSFA$ constructs the DFA over $\Gamma_{\aut{M}}$ where concrete transitions agree with symbolic transitions of the original SFA.
     See Figure~\ref{fig:DFA} for the resulting DFA. 
     Then, using $\CharDFA$, it returns the sample set: $$\aut{S}_\aut{M} = \{ \langle \epsilon, 0  \rangle, \langle 0, 1  \rangle, \langle 100, 0  \rangle, \langle 200, 0  \rangle,
     \langle 0\cdot 0, 1  \rangle, \langle 0\cdot 100, 1  \rangle,\langle 0\cdot 200, 0  \rangle
    \}.$$

    Now, assume algorithm $\InferSFA$ is given the set 
    $\aut{S} = \aut{S}_\aut{M} \cup \{ \langle  150, 0 \rangle \}$
    over the alphabet $\Sigma = \{ 0, 100, 150, 200 \}$. 
    The algorithm applies $\Decontaminate_{\class{M}_{{\alge{A_\dom{N}}}}}$ that generates the set $\aut{S'}$ over $\Gamma_{\aut{M}}$. %, where $\aut{S'}$ is defined as follows. 
    To do so, 
    it first finds the set $\Gamma_{\aut{M}}$
    of all elements that are a minimal left point of some interval, and then chooses from $\aut{S}$ the words over $\Gamma_{\aut{M}}$. 
    It does so as follows. First, note that $100 \sim_\aut{S} 150$, $100 \sim_\aut{S} 200$ and $150 \sim_\aut{S} 200$, 
    while $0 \not\sim_\aut{S} 100, 150, 200$. Since $0$ is the minimal element it has to be in $\Gamma_{\aut{M}}$; and since $100$ is the minimal element that is not equivalent to $0$ it has to define a new interval and thus is in $\Gamma_{\aut{M}}$ as well.    Next, we consider suffixes of words over $\{0, 100 \}$. These are $0\cdot 0$ and $0\cdot 100$ 
    which are equivalent, and 
    $0\cdot 200$ which is not equivalent to the former. Since $100$ is equivalent to $0$ it does not define a new interval now, but $200$ does as it is the minimal (and only) element that is not equivalent to $0$ when considering suffixes of $0$. Then, we deduce that 
    $\Gamma_{\aut{M}} = \{0, 100, 200  \}$
    and thus 
 $\aut{S'} =
 \{ \langle \epsilon, 0  \rangle, \langle 0, 1  \rangle, \langle 100, 0  \rangle,  \langle 200, 0  \rangle, 
     \langle 0\cdot 0, 1  \rangle, \langle 0\cdot 100, 1  \rangle, \langle 0\cdot 200, 0  \rangle
    \}$.

    Now Algorithm $\InferSFA$ is applied to the set $\aut{S'}$ 
    and the resulting DFA would be 
    the DFA $\aut{D}_\aut{M}$ of Figure~\ref{fig:DFA}. 
    Then it applies $\generalize_{\class{M}_{{\alge{A_\dom{N}}}}}$ described in Example~\ref{ex:con_gen} and the result will be the original SFA of Figure~\ref{fig:SFA}.  
    That is, for outgoing transitions of $q_0$ it applies $\generalize_{\alge{A_\dom{N}}} (\partition{ \{ 0 \}, \{ 100, 200\} }) =  \partition{[0, 100), [100, \infty) }$ and for outgoing transitions of~$q_1$ it applies $\generalize_{\alge{A_\dom{N}}} (\partition{ \{ 0, 100 \}, \{ 200 \} }) =  \partition{[0, 200), [200, \infty) }$ and uses these predicates to annotate the corresponding transitions in the SFA. 
\end{exa}

\begin{figure}[t]
\begin{tikzpicture}[scale=0.14]
\tikzstyle{every node}+=[inner sep=0pt]
\draw [black] (18.3,-19) circle (3);
\draw (18.3,-19) node {$q_0$};
\draw [black] (32.4,-19) circle (3);
\draw (32.4,-19) node {$q_1$};
\draw [black] (32.4,-19) circle (2.4);
\draw [black] (10.8,-19) -- (15.3,-19);
\fill [black] (15.3,-19) -- (14.5,-18.5) -- (14.5,-19.5);
\draw [black] (19.952,-16.519) arc (135.06693:44.93307:7.625);
\fill [black] (30.75,-16.52) -- (30.54,-15.6) -- (29.83,-16.31);
\draw (25.35,-13.78) node [above] {$0$};
\draw [black] (33.36,-16.17) arc (189:-99:2.25);
\draw (37.7,-13) node [right] {$0, 100$};
\fill [black] (35.23,-18.04) -- (36.1,-18.41) -- (35.94,-17.42);
\draw [black] (18.597,-21.973) arc (33.44395:-254.55605:2.25);
\draw (9,-22.77) node [below] {$100, 200$};
\fill [black] (16.12,-21.04) -- (15.18,-21.07) -- (15.73,-21.9);
\draw [black] (30.128,-20.939) arc (-58.82797:-121.17203:9.231);
\fill [black] (20.57,-20.94) -- (21,-21.78) -- (21.52,-20.93);
\draw (25.35,-22.77) node [below] {$200$};
\end{tikzpicture}     \caption{The DFA $\aut{D}_\aut{M}$ constructed in \CharSFA}
        \label{fig:DFA}
    
\end{figure}

%\section{Example for the Need of \Decontaminate} 

\begin{figure}[t]%[!htb]
    \begin{minipage}{.5\textwidth}
            \centering
    		{\small
        \begin{tikzpicture}[scale=0.15]
\tikzstyle{every node}+=[inner sep=0pt]
\draw [black] (26,-24.5) circle (3);
\draw (26,-24.5) node {$q_\iota$};
\draw [black] (36.3,-19.3) circle (3);
\draw (36.3,-19.3) node {$q_1$};
\draw [black] (36.3,-19.3) circle (2.4);
\draw [black] (36.3,-29) circle (3);
\draw (36.3,-29) node {$q_2$};
\draw [black] (36.3,-29) circle (2.4);
\draw [black] (45.1,-24.5) circle (3);
\draw (45.1,-24.5) node {$q_3$};
\draw [black] (17.8,-24.5) -- (23,-24.5);
\fill [black] (23,-24.5) -- (22.2,-24) -- (22.2,-25);
\draw [black] (26.845,-21.648) arc (150.66009:82.91416:6.691);
\fill [black] (33.5,-18.29) -- (32.77,-17.69) -- (32.65,-18.68);
\draw (26.27,-18.44) node [above] {$[0,100)$};
\draw [black] (33.412,-29.734) arc (-87.35775:-139.84263:7.391);
\fill [black] (33.41,-29.73) -- (32.59,-29.27) -- (32.64,-30.27);
\draw (24.92,-29.68) node [below] {$[100,\mbox{ }\infty)$};
\draw [black] (39.174,-18.56) arc (89.95337:28.88818:5.931);
\fill [black] (44.36,-21.63) -- (44.41,-20.68) -- (43.54,-21.17);
\draw (47.51,-18.88) node [above] {$[100,\mbox{ }\infty)$};
\draw [black] (34.977,-16.62) arc (234:-54:2.25);
\draw (36.3,-12.05) node [above] {$[0,100)$};
\fill [black] (37.62,-16.62) -- (38.5,-16.27) -- (37.69,-15.68);
\draw [black] (37.623,-31.68) arc (54:-234:2.25);
\draw (36.3,-36.25) node [below] {$[100,\mbox{ }\infty)$};
\fill [black] (34.98,-31.68) -- (34.1,-32.03) -- (34.91,-32.62);
\draw [black] (43.555,-27.046) arc (-43.25739:-82.57551:7.156);
\fill [black] (43.56,-27.05) -- (42.64,-27.29) -- (43.37,-27.97);
\draw (44.99,-29.03) node [below] {$[0,100)$};
\draw [black] (47.438,-22.638) arc (156.26477:-131.73523:2.25);
\draw (52.44,-22.7) node [right] {$[0,\infty)$};
\fill [black] (48,-25.22) -- (48.53,-26) -- (48.93,-25.09);
\end{tikzpicture}
        } \captionof{figure}{The SFA $\aut{M}_1$}\label{fig:SFA_decon}
    \end{minipage}%
    %\captionof{algorithm}{concretize}
    \hfill
    \begin{minipage}{.5\textwidth}
            \centering
    		
    		{\small 
    		\begin{tikzpicture}[scale=0.12]
\tikzstyle{every node}+=[inner sep=0pt]
\draw [black] (36.6,-20.5) circle (3);
\draw [black] (24.8,-30.4) circle (3);
\draw [black] (24.8,-30.4) circle (2.4);
\draw [black] (36.6,-30.4) circle (3);
\draw [black] (36.6,-30.4) circle (2.4);
\draw [black] (47.4,-30.4) circle (3);
\draw [black] (47.4,-30.4) circle (2.4);
\draw [black] (15.7,-40.4) circle (3);
\draw [black] (15.7,-40.4) circle (2.4);
\draw [black] (24.8,-40.4) circle (3);
\draw [black] (24.8,-50.9) circle (3);
\draw [black] (32.1,-40.4) circle (3);
\draw [black] (32.1,-40.4) circle (2.4);
\draw [black] (41.2,-40.4) circle (3);
\draw [black] (54.4,-39) circle (3);
\draw [black] (54.4,-39) circle (2.4);
\draw [black] (34.3,-22.43) -- (27.1,-28.47);
\fill [black] (27.1,-28.47) -- (28.03,-28.34) -- (27.39,-27.57);
\draw (29.69,-24.96) node [above] {$0$};
\draw [black] (36.6,-23.5) -- (36.6,-27.4);
\fill [black] (36.6,-27.4) -- (37.1,-26.6) -- (36.1,-26.6);
\draw (36.1,-25.45) node [left] {$100$};
\draw [black] (38.81,-22.53) -- (45.19,-28.37);
\fill [black] (45.19,-28.37) -- (44.94,-27.46) -- (44.26,-28.2);
\draw (44.01,-24.96) node [above] {$150$};
\draw [black] (22.78,-32.62) -- (17.72,-38.18);
\fill [black] (17.72,-38.18) -- (18.63,-37.93) -- (17.89,-37.25);
\draw (19.71,-33.94) node [left] {$0$};
\draw [black] (24.8,-33.4) -- (24.8,-37.4);
\fill [black] (24.8,-37.4) -- (25.3,-36.6) -- (24.3,-36.6);
\draw (25.3,-35.4) node [right] {$100$};
\draw [black] (24.8,-43.4) -- (24.8,-47.9);
\fill [black] (24.8,-47.9) -- (25.3,-47.1) -- (24.3,-47.1);
\draw (24.3,-45.65) node [left] {$0$};
\draw [black] (36.6,-11.9) -- (36.6,-17.5);
\fill [black] (36.6,-17.5) -- (37.1,-16.7) -- (36.1,-16.7);
\draw [black] (35.37,-33.14) -- (33.33,-37.66);
\fill [black] (33.33,-37.66) -- (34.12,-37.14) -- (33.2,-36.73);
\draw (34.63,-34.4) node [left] {$100$};
\draw [black] (37.85,-33.13) -- (39.95,-37.67);
\fill [black] (39.95,-37.67) -- (40.07,-36.74) -- (39.16,-37.16);
\draw (39.62,-34.37) node [right] {$0$};
\draw [black] (49.29,-32.73) -- (52.51,-36.67);
\fill [black] (52.51,-36.67) -- (52.39,-35.74) -- (51.61,-36.37);
\draw (50.34,-36.13) node [left] {$250$};
\end{tikzpicture}

    		}
    	\captionof{figure}{The prefix-tree automaton for the set $\aut{S}'$ of Example~\ref{ex:decontanimate}}\label{fig:prefix_tree}
    \end{minipage}
    %\captionof{algorithm}{generalize}
\end{figure}

\begin{figure}[t!]

\begin{tabular}{cc}
%{R}{0.4\textwidth}
        %\vspace{-6mm}
        %\scalebox{.8}
        {\small
        \begin{tikzpicture}[scale=0.15]
\tikzstyle{every node}+=[inner sep=0pt]
\draw [black] (26,-24.5) circle (3);
\draw (26,-24.5) node {$q_\iota$};
\draw [black] (36.3,-19.3) circle (3);
\draw (36.3,-19.3) node {$q_1$};
\draw [black] (36.3,-19.3) circle (2.4);
\draw [black] (36.3,-29) circle (3);
\draw (36.3,-29) node {$q_2$};
\draw [black] (36.3,-29) circle (2.4);
\draw [black] (45.1,-24.5) circle (3);
\draw (45.1,-24.5) node {$q_3$};
\draw [black] (17.8,-24.5) -- (23,-24.5);
\fill [black] (23,-24.5) -- (22.2,-24) -- (22.2,-25);
\draw [black] (26.845,-21.648) arc (150.66009:82.91416:6.691);
\fill [black] (33.5,-18.29) -- (32.77,-17.69) -- (32.65,-18.68);
\draw (26.27,-18.44) node [above] {$0, 150$};
\draw [black] (33.412,-29.734) arc (-87.35775:-139.84263:7.391);
\fill [black] (33.41,-29.73) -- (32.59,-29.27) -- (32.64,-30.27);
\draw (24.92,-29.68) node [below] {$100$};
\draw [black] (39.174,-18.56) arc (89.95337:28.88818:5.931);
\fill [black] (44.36,-21.63) -- (44.41,-20.68) -- (43.54,-21.17);
\draw (47.51,-18.88) node [above] {$100$};
\draw [black] (34.977,-16.62) arc (234:-54:2.25);
\draw (36.3,-12.05) node [above] {$0, 250$};
\fill [black] (37.62,-16.62) -- (38.5,-16.27) -- (37.69,-15.68);
\draw [black] (37.623,-31.68) arc (54:-234:2.25);
\draw (36.3,-36.25) node [below] {$100$};
\fill [black] (34.98,-31.68) -- (34.1,-32.03) -- (34.91,-32.62);
\draw [black] (43.555,-27.046) arc (-43.25739:-82.57551:7.156);
\fill [black] (43.56,-27.05) -- (42.64,-27.29) -- (43.37,-27.97);
\draw (44.99,-29.03) node [below] {$0$};
\draw [black] (47.438,-22.638) arc (156.26477:-131.73523:2.25);
\draw (52.44,-22.7) node [right] {$0, 100, 150$};
\fill [black] (48,-25.22) -- (48.53,-26) -- (48.93,-25.09);
\end{tikzpicture}
        }
        
             & 
             
    {\small
        \begin{tikzpicture}[scale=0.15]
\tikzstyle{every node}+=[inner sep=0pt]
\draw [black] (26,-24.5) circle (3);
\draw (26,-24.5) node {$q_\iota$};
\draw [black] (36.3,-19.3) circle (3);
\draw (36.3,-19.3) node {$q_1$};
\draw [black] (36.3,-19.3) circle (2.4);
\draw [black] (36.3,-29) circle (3);
\draw (36.3,-29) node {$q_2$};
\draw [black] (36.3,-29) circle (2.4);
\draw [black] (45.1,-24.5) circle (3);
\draw (45.1,-24.5) node {$q_3$};
\draw [black] (17.8,-24.5) -- (23,-24.5);
\fill [black] (23,-24.5) -- (22.2,-24) -- (22.2,-25);
\draw [black] (26.845,-21.648) arc (150.66009:82.91416:6.691);
\fill [black] (33.5,-18.29) -- (32.77,-17.69) -- (32.65,-18.68);
\draw (26.27,-18.44) node [above] {$0$};
\draw [black] (33.412,-29.734) arc (-87.35775:-139.84263:7.391);
\fill [black] (33.41,-29.73) -- (32.59,-29.27) -- (32.64,-30.27);
\draw (24.92,-29.68) node [below] {$100, 150$};
\draw [black] (39.174,-18.56) arc (89.95337:28.88818:5.931);
\fill [black] (44.36,-21.63) -- (44.41,-20.68) -- (43.54,-21.17);
\draw (47.51,-18.88) node [above] {$100$};
\draw [black] (34.977,-16.62) arc (234:-54:2.25);
\draw (36.3,-12.05) node [above] {$0$};
\fill [black] (37.62,-16.62) -- (38.5,-16.27) -- (37.69,-15.68);
\draw [black] (37.623,-31.68) arc (54:-234:2.25);
\draw (36.3,-36.25) node [below] {$100, 250$};
\fill [black] (34.98,-31.68) -- (34.1,-32.03) -- (34.91,-32.62);
\draw [black] (43.555,-27.046) arc (-43.25739:-82.57551:7.156);
\fill [black] (43.56,-27.05) -- (42.64,-27.29) -- (43.37,-27.97);
\draw (44.99,-29.03) node [below] {$0$};
\draw [black] (47.438,-22.638) arc (156.26477:-131.73523:2.25);
\draw (52.44,-22.7) node [right] {$0, 100, 150$};
\fill [black] (48,-25.22) -- (48.53,-26) -- (48.93,-25.09);
\end{tikzpicture}
        }

\end{tabular}

        \caption{Two DFAs that are consistent with the set $\aut{S}'$ of Example~\ref{ex:decontanimate}.}
        %\vspace{10cm}
        \label{fig:two_DFAs}
         %\vspace{-6mm}
\end{figure}

 \begin{exa} \label{ex:decontanimate}\label{example:decontaminate}
Consider the SFA $\aut{M}_1$ of Figure~\ref{fig:SFA_decon}. 
Applying $\CharDFA(\Concretize_{\class{M}_{\alge{A}_\dom{N}}} \!(\aut{M}_1))$ results in the following $\aut{S}$. 
$$\aut{S} = \{ \langle \epsilon, 0 \rangle, \langle 0, 1 \rangle, \langle 100, 1 \rangle, \langle 0\cdot 0, 1  \rangle, \langle 0\cdot 100, 0 \rangle, \langle 100\cdot 100, 1 \rangle, \langle  100\cdot 0, 0\rangle, \langle 0\cdot100\cdot0, 0 \rangle \}$$

Now, let $\aut{S}'= \aut{S} \cup \{ \langle 150, 1 \rangle, \langle 150\cdot 250, 1 \rangle  \}$. Note that $\aut{S}'$
is consistent with $\aut{L}(\aut{M}_1)$. 
When trying to learn an SFA from $\aut{S}'$ and applying $\InferDFA(\aut{S}')$, the algorithm cannot distinguish between the words $150$ and $100$, as well as between $0$ and $150$. 
 The same holds also for $0\cdot 0$ vs.\ $150\cdot 250$, and $100\cdot 100$ vs.\ $150\cdot 250$.

Figure~\ref{fig:two_DFAs} presents two DFAs that are both consistent with the set $\aut{S}'$. Since $\aut{S}'$ does not contain any characteristic sample for DFAs over the alphabet $\{ 0, 100, 150, 250 \}$, \InferDFA\ concludes that $\aut{S}'$ does not subsume any characteristic sample, and returns the prefix-tree automaton, given in Figure~\ref{fig:prefix_tree}.

This example illustrates that {\InferSFA} cannot classify $150$ and $150\cdot 250$ by only applying \InferDFA, without reasoning about the predicates of the algebra. To this end we provide the function $\Decontaminate_{\class{M}_\alge{A}}$, which is able, in the case of a monotonic algebra, to find which letters are the ones that should be used to define new predicates.
  \end{exa}

%\sub
\section{Negative Result}\label{subsec:negative}

The result of Theorem~\ref{prop:sample_monotonic} 
does not extend to the non-monotonic case, as~stated in Theorem~\ref{prop:no_poly_set} regarding SFAs over the general propositional algebra.
Let $\dom{D}_\mathbb{B}=\{ {\mathbb{B}^k} \}_{k\in\mathbb{N}}$. Recall that $\mathbb{B}=\{0,1\}$ and $\mathbb{B}^k$ is the set of all valuations of $k$ atomic propositions.
Let $\mathbb{P}_\mathbb{B}=\{ \mathbb{P}_{\mathbb{B}_k} \}_{k\in\mathbb{N}}$ 
where  $\mathbb{P}_{\mathbb{B}_k}$ is the set of predicates over at most $k$ atomic propositions.
Let~$\alge{A}_\mathbb{B}$ be the Boolean algebra defined over the discrete domain $\dom{D}_\mathbb{B}$ and
the set of predicates $\mathbb{P}_\mathbb{B}$, and the usual operators $\vee$, $\wedge$ and $\neg$.
Let $\class{M}_{\alge{A}_\mathbb{B}}$ be the class of SFAs over the Boolean algebra~$\alge{A}_\mathbb{B}$. We show that unless $P=NP$, this class of SFAs is not efficiently identifiable.\footnote{
This result may be contrasted with~\cite{ArgyrosD18} who provide a positive learnability result regarding SFAs over the OBDDs algebra. The result of~\cite{ArgyrosD18} is with respect to query learning, while Theorem~\ref{prop:no_poly_set} concerns efficient identifiability in the limit. As we discuss in \autoref{sec:discuss}, one cannot derive {efficient learnability} from a positive result in the query learning setting. Moreover, in Theorem~\ref{prop:no_poly_set} (as well as in Corollary~\ref{cor:neg-qlearn-sfa-prop-alg} and the discussion in \autoref{sec:query_learning}) we refer to efficient learnability with respect the propositional algebra as defined in~\autoref{sec:bool-alge} where the size is measured with respect to the number of atomic propositions, while~\cite{ArgyrosD18} refer to the size of the SFAs in which the predicates are OBDDs, whose size is measured by the number of nodes in the OBDD. However, the number of nodes in an OBDD can be exponential in the number of atomic propositions. Therefore, our result has no conflict with the result of~\cite{ArgyrosD18}.}

\begin{thm}
\label{prop:no_poly_set}
    The class $\class{M}_{\alge{A}_\mathbb{B}}$ is not efficiently identifiable unless $P=NP$. 
\end{thm}

\begin{proof}

We show that there is no pair of efficient dyadic concretizing and generalizing functions $f_c:\ppart{\class{P}_\mathbb{B}}\rightarrow \cpart{\dom{D}_\mathbb{B}}$ and  $f_g: \cpart{\dom{D}_\mathbb{B}}\rightarrow \ppart{\class{P}_\mathbb{B}}$
    unless $P=NP$. From Theorem~\ref{theorem:necessary-cond} it follows that $\class{M}_{\mathbb{B}}$ is not efficiently identifiable unless $P=NP$.

Assume towards contradiction that such a pair of functions exist. We provide a polynomial
time algorithm $\Asat$ for SAT. On a predicate $\varphi$, the algorithm $\Asat$ invokes
$f_c(\partition{\varphi,\neg\varphi})$. Suppose the returned concrete partition is $\partition{\Gamma_1,\Gamma_2}$.
Then $\Asat$ returns ``true'' if and only if $\Gamma_1\neq \emptyset$.
Correctness follows from the fact that if there exists a system of characteristic samples for $\mathbb{P}_\mathbb{B}$
then the set of positive examples associated with a satisfiable predicate $\varphi$ must be non-empty,
as otherwise $f_g$ cannot distinguish $\varphi$ from $\bot$.
\end{proof}

\section{Query Learning}\label{sec:query_learning}
%\input{query}

%\vspace{-2mm}
The paradigm of \emph{query learning} stipulates that the \emph{learner} can interact with an \emph{oracle} (\emph{teacher})
by asking it several types of allowed queries. 
In this section we consider these queries to be \emph{membership queries} (\mq) and \emph{equivalence queries} (\eq). We say that a class~$\class{M}$ of automata is efficiently learnable using \mq s / \eq s / both \mq s and \eq s if there is an algorithm that for every language $L$ with a representation in $\class{M}$
asks a polynomial number of \mq s / \eq s / both \mq s and \eq s and outputs an automaton in $\class{M}$ that is polynomial in the minimal representation of $L$ in $\class{M}$. 

Angluin showed, on the negative side,  that regular languages cannot be efficiently learned (in the exact model) from only \mq s~\cite{Angluin81} or only \eq s~\cite{Angluin90}.
On the positive side, she showed that regular languages, represented as DFAs, can be efficiently learned
using both \mq s and \eq s~\cite{Angluin87}. 
The celebrated algorithm, termed \lstar, was extended
to learning many other classes of languages and representations, e.g., \cite{Sakakibara90,BV96,AartsV10,BolligHKL09,AngluinEF15,MalerP95,AngluinF16,NitayFZ21}. See the survey~\cite{Fisman18} for more references.

In particular, an extension of \lstar, termed \matstar, to learn SFAs was provided in~\cite{ArgyrosD18} which proved that SFAs over an algebra \alge{A} can be efficiently learned using \matstar\ if and only if the underlying algebra is efficiently learnable, and the size of disjunctions  of~$k$ predicates does not grow exponentially in~$k$.\footnote{As is the case, for instance, in the OBDD (Ordered Binary Decisions Diagrams) algebra~\cite{DBLP:journals/tc/Bryant86}.}
From this it was concluded that SFAs over the following underlying algebras are efficiently learnable: Boolean algebras over finite domains, equality algebra,  tree automata algebra, and SFAs algebra.  
 Efficient learning of SFAs over a monotonic algebra using \mq s and \eq s was established in~\cite{ChubachiDYS17}, which improved the results 
of~\cite{MalerM14,MalerM17} by using a binary search instead of a helpful teacher.

The result of~\cite{ArgyrosD18} provides means to establish new positive results on learning classes of SFAs using \mq s and \eq s, but it 
does not provide means for obtaining negative results for query learning of SFAs using \mq s and \eq s.
We strengthen this result by providing a learnability result that is independent of the use of a specific learning algorithm.
In particular, we show that efficient learnability of a Boolean algebra $\alge{A}$ using \mq s and \eq s is a necessary condition for the learnability of a class of SFAs over $\alge{A}$,
as we state in Theorem~\ref{thm:neg-query}.

\begin{thm}\label{thm:neg-query}
    A class of  SFAs $\class{M}$ over a Boolean algebra $\alge{A}$, that contains all basic SFAs over $\alge{A}$, is polynomially learnable using \mq s and \eq s,
    only if  $\alge{A}$ is polynomially learnable using \mq s and \eq s.
\end{thm}

\begin{proof}
    Assume that $\class{M}$ is polynomially learnable using \mq s and \eq s, using an algorithm~$\algoQSFA$.
    We show that there exists a polynomial learning algorithm~$\algoQalge$ for the algebra~$\alge{A}$ using \mq s and \eq s.
    The algorithm~$\algoQalge$ uses~$\algoQSFA$ as a subroutine, and behaves as a teacher for~$\algoQSFA$. 
    Whenever $\algoQSFA$ asks an
    $\class{M}$-$\mq$ on word $\gamma_1\ldots\gamma_k$, if $k>1$ then $\algoQalge$ answers ``no''.
    If~$k\!=\!1$ then the~$\class{M}$-$\mq$ is essentially an $\alge{A}$-$\mq$, thus 
    $\algoQalge$ issues this query  and passes the answer to $\algoQSFA$.
        Whenever $\algoQSFA$ asks a
    $\class{M}$-$\eq$  on SFA $\aut{M}$, if $\aut{M}$ is not of the form $\aut{M}_\psi$ for some $\psi$ (as defined in Definition~\ref{def:non-trivial})
    then $\algoQalge$ answers ``no'' to the $\class{M}$-$\eq$
    and returns some word $w\in\lang{\aut{M}}$ s.t.\ $|w|>1$ and $w$ was not provided before, as a counterexample. To this aim it can record the largest counterexample given so far (according to the lexicographic order) and return the next one in this order. 
    Otherwise (if the SFA is of the form $\aut{M}_\psi$ for some~$\psi$) $\algoQalge$ asks an $\alge{A}$-$\eq$ on $\psi$. 
    If the answer is ``yes'' then $\algoQalge$ terminates and returns $\psi$ as the result of the learning algorithm; if the answer to the $\alge{A}$-$\eq$ on $\psi$ is ``no'', then the provided counterexample $\langle \gamma,b_\gamma \rangle$ is passed back to $\algoQSFA$ together with the answer ``no'' to the $\class{M}$-$\eq$. It is easy to verify that $\algoQalge$ terminates correctly in polynomial time.
    \end{proof}

%\vspace{-3mm}
From Theorem~\ref{thm:neg-query} we derive  what we believe to be the first negative result on learning SFAs from \mq s and \eq s, as we show
that SFAs over $\alge{A}_{\class{B}_k}$, the propositional algebra over~$k$ variables, are not polynomially learnable using \mq s and \eq s.
Polynomiality is measured with respect to the parameters $\langle n,m,l\rangle$ representing the size of the SFA and the number $k$ of atomic propositions.
Note that the algebra $\alge{A}_{\class{B}_k}$ is a restriction of the algebra $\alge{A}_\mathbb{B}$ considered in \autoref{subsec:negative} and therefore implies a negative result also with regard to the algebra $\alge{A}_\mathbb{B}$
considered there.

We achieve this by showing that no learning algorithm $\algor{A}$ for the propositional algebra  using \mq s and \eq s can do better than asking $2^k$ \mq s/\eq s, where $k$ is the number of atomic propositions.\footnote{In~\cite{Nakamura00} Boolean formulas represented using OBDDs are claimed to be polynomially learnable with \mq s and \eq s. However,~\cite{Nakamura00} measures the size of an OBDD by its number of nodes, which can be exponential in the number of propositions.} 
We assume the learning algorithm is \emph{sound},
that is, if  $S_i^+$ and $S_i^-$ are the sets of positive and negative examples observed by the algorithm up to stage $i$, then at stage~$i+1$ the algorithm will not ask a \mq\ for a word in $S_i^+ \cup S_i^-$ 
or an \eq\ for an automaton that rejects a word in $S_i^+$ or accepts a word in~$S_i^-$.

\begin{prop} \label{proposition:propositional algebra learning}
	Let $\algor{A}$ be a sound learning algorithm for the propositional algebra over~$\mathbb{B}^k$. There exists a target predicate $\psi$ of size $k$, for which $\algor{A}$ will be forced to ask at least~$2^k-1$ queries (either \mq\ or \eq).
\end{prop}

\begin{proof}
	Since $\algor{A}$ is sound, at stage $i+1$ we have $S_{i+1}^+\supseteq S_i^+$ and $S_{i+1}^-\supseteq S_i^-$ and at least~one inclusion is strict. Since the size of the concrete alphabet is $2^k$, 
	for every round $i<2^k$,
	an adversarial teacher can answer both \mq s and~\eq s negatively. In the case of an \eq\ there must be an element in $\mathbb{B}^k\setminus (S_i^-\cup S_i^+)$ with which the provided automaton disagrees. The adversary~will return one such element as a counterexample. This forces $\algor{A}$ to ask at least $2^k\!-\!1$ queries.
	Note that for any element $v$ in $\mathbb{B}^k$ there exists a predicate $\varphi_v$ of size $k$ such that $\sema{\varphi_v}=\{v\}$.
\end{proof}

\begin{cor}\label{cor:neg-qlearn-sfa-prop-alg}
SFAs  over the propositional algebra $\alge{A}_{\class{B}_k}$ with $k$ propositions cannot be learned in $\mathit{poly}(k)$ time using \mq s and \eq s. 
\end{cor}

  The propositional algebra $\alge{A}_{\class{B}_k}$ is a special case of 
  the $n$-dimensional boxes algebra. 
  Learning $n$-dimensional boxes was studied
 using \mq s and \eq s~\cite{DBLP:conf/colt/GoldbergGM94,DBLP:journals/siamcomp/BshoutyGGM98,DBLP:journals/algorithmica/BeimelK98}, as well as in the PAC setting~\cite{DBLP:journals/ipl/BeimelK00}. 
 The algorithms presented in~\cite{DBLP:conf/colt/GoldbergGM94,DBLP:journals/siamcomp/BshoutyGGM98,DBLP:journals/algorithmica/BeimelK98,DBLP:journals/ipl/BeimelK00} are mostly exponential in $n$. Alternatively,~\cite{DBLP:conf/colt/GoldbergGM94,DBLP:journals/siamcomp/BshoutyGGM98} suggest algorithms that are exponential in the number of boxes in the union. In~\cite{DBLP:journals/algorithmica/BeimelK98} a linear query learning algorithm for unions of disjoint boxes is presented. 
 Since $n$-dimensional boxes subsume the propositional algebra, Corollary~\ref{cor:neg-qlearn-sfa-prop-alg} implies the following.  
  \begin{cor}
The class of SFAs over the $n$-dimensional boxes algebra cannot be learned in $\mathit{poly}(n)$ time using \mq s and \eq s. 
\end{cor}

\section{Discussion}\label{sec:discuss}
We examined the question of learnability of a class of SFAs over certain algebras
where the main focus of our study is on passive learning. 
We provided a necessary condition for identification of SFAs 
in the limit using polynomial time and data, as well as a necessary condition for efficient learning of SFAs using membership and equivalence queries.
We note that a positive result on learning deterministic SFAs using \mq s and \eq s implies a positive result 
for identification of deterministic SFAs 
in the limit using polynomial time and data. 
The latter follows because a systematic set of characteristic samples $\{S_L\}_{L\in\class{L}}$ for a class of languages $\class{L}$ may be obtained by collecting
the words observed by the query learner when learning $L$, and given the SFA is deterministic, the words in the sample can be restricted to ones of polynomial size, thus the size of the sample is  polynomial in the size of the SFA.\footnote{For non-deterministic SFAs, as is the case for NFAs, it might be the case that some exponentially long words are required to learn the language~\cite{Higuera97}.} However, it does not imply a positive result regarding the stronger notion of {efficient identifiability}, as the latter requires the set to be also constructed efficiently, and the complexity analysis for query learning does not include the complexity of the teacher in computing queries, e.g., deciding equivalence and in constructing counterexamples.
We thus 
provided a
sufficient condition for {efficient identification} of a class of SFAs, and showed that the class of SFAs over any monotonic algebra satisfies these conditions.

We hope that these sufficient or necessary conditions will help to obtain more positive and negative results for learning of SFAs,
and spark an interest in investigating characteristic samples in other automata models used in verification.

\section*{Acknowledgments}
This research was partially supported by the Israel Science Foundation grant 2507/21.
S.~Zilles acknowledges financial support by the Natural Sciences and Engineering Research Council (NSERC) of Canada, both through the Canada Research Chairs program and the Discovery Grants program. She was also supported through the Canada CIFAR AI Chairs program as an Affiliate Chair with the Alberta Machine Intelligence Institute (Amii).

\bibliographystyle{alphaurl}
\bibliography{bib}

\newpage
\appendix

\section{Efficient Identification of DFAs}\label{App:DFA}
It was shown by~\cite{DBLP:journals/iandc/Gold78,RPNI} that  DFAs are identifiable in the limit using polynomial time and data.
Since the proof of Proposition~\ref{lemma:sufficient_monotonic}  relies on some properties of the involved procedures and 
for completeness of the presentation, we provide a complete description of the procedures showing that 
DFAs are identifiable in the limit using polynomial time and data, and that they satisfy the required properties.

\bigparagraph{Theorem~\ref{prop:charDFA} (restated) \cite{RPNI}}
    \emph{%\begin{theorem}[\cite{RPNI}]%\label{prop:charDFA}
    \itemI
    The class of DFAs is efficiently identifiable via
    procedures $\CharDFA$ and $\InferDFA$. 
    \itemII
    Furthermore, these procedures satisfy that if 
    $\aut{D}$ is a minimal {and complete} DFA and
    $\CharDFA(\aut{D})=\mathcal{S}_\aut{D}$ then
     the following holds:
   \begin{enumerate}
          \item %\label{item:lex}
          $\mathcal{S}_\aut{D}$ contains a prefix-closed set $A$ of access words. 
          Moreover,  $A$ can be chosen to contain only lex-access words,  i.e., only the lexicographically smallest access word for each state. 
        \item For every $u_1,u_2\in A$ it holds that $u_1\not\sim_{\aut{S}_\aut{D}} u_2$. %\label{item:accesswords}
        \item 
        For every $u,v\in A$ and $\sigma\in\Sigma$, if $\Delta(q_\iota, u\sigma)\neq \Delta(q_\iota, v)$ then $u\sigma\not\sim_{\aut{S}_\aut{D}}v$. %\label{item:notsim_inS}
    \end{enumerate}
%\end{theorem}
}\quad 

To prove Theorem~\ref{prop:charDFA} we first show, in \autoref{sec:charDFA},
that given a
DFA $\aut{D} = \langle \Sigma,Q,q_\iota ,F,\Delta \rangle$
we can construct a polynomial-sized sample of words $\mathcal{S}_\aut{D}$ that
agrees with $\aut{D}$ and satisfies the required properties.
In \autoref{sec:inferDFA} we 
show an algorithm that (i) can infer in polynomial time from a given sample $\aut{S}$
a DFA that agrees with $\aut{S}$, and (ii) if it is given the set $\mathcal{S}_\aut{D}$, or
any set $\mathcal{S}\supseteq\mathcal{S}_\aut{D}$ that
agrees with $\aut{D}$, then it infers a DFA that is equivalent to $\aut{D}$.
All this together proves Theorem~\ref{prop:charDFA} (and explains why
we can refer to $\mathcal{S}_\aut{D}$ as the \emph{characteristic sample}).

\subsection{Constructing a characteristic set}\label{sec:charDFA}

The algorithm $\CharDFA$ works as follows. It first creates a prefix-closed set of access words to states.
This can  be done by considering the graph of the automaton and running an algorithm for finding a spanning tree  from
the initial state. Choosing one of the letters on each edge, the access word for a state is obtained by concatenating
the labels on the unique path of the obtained tree that reaches that state.
If we wish to work with lex-access words, we can use a depth-first search algorithm that spans branches according to the order
of letters in $\Sigma$, starting from the smallest. The labels  on the paths of the spanning tree constructed this way will form the set of lex-access words.

Let $S$ be the set of access words (or lex-access words).
Next the algorithm turns to find a distinguishing word $v_{i,j}$ for every pair of state $s_i,s_j\in S$ (where $s_i\neq s_j$).
Lemma~\ref{lem:distinquishing-word-dfa} below states that any pair of states of the minimal DFA has a distinguishing word of size quadratic in the size of the DFA.
Let $E$
be the set of all such distinguishing words. We may assume $\epsilon\in E$.\footnote{Unless $\aut{D}$ accepts all words or rejects all words, it has at least one  accepting state and one rejecting state, and $\epsilon$ is the shortest word distinguishing these states. If 
all states of $\aut{D}$ are accepting (or all rejecting) the algorithm returns $\mathcal{S}_\aut{D}=\{\langle  \epsilon, 1 \rangle\}$ (resp.  $\mathcal{S}_\aut{D}=\{\langle  \epsilon, 0 \rangle\}$).}
	Then the algorithm returns the set $\mathcal{S}_\aut{D}=\{\langle w, \aut{D}(w) \rangle ~|~ w\in (S\cdot E) \cup (S \cdot \Sigma \cdot E)\}$ where $\aut{D}(w)$ is the label $\aut{D}$ gives $w$ (i.e., $1$ if it is accepted,  and $0$ otherwise). 
	
	It is easy to see that $\mathcal{S}_\aut{D}$ satisfies the properties of Theorem~\ref{prop:charDFA}.

\begin{lem}\label{lem:distinquishing-word-dfa}
	Let $\aut{D} = \langle \Sigma,Q,q_\iota ,F,\Delta \rangle $ be a minimal DFA, and let $q_1,q_2\in Q$ s.t. $q_1\neq q_2$.
	There exists a polynomial time procedure that returns a word $v$ of size at most $|Q|^2$ such that $\Delta(q_1,v)$ is accepting iff $\Delta(q_2,v)$ is rejecting.
\end{lem}
\begin{proof}
	We can apply the product construction to $\aut{D}_i = \langle \Sigma,Q,q_i ,F,\Delta \rangle $ for $i\in\{1,2\}$ 
	and search for a path from the initial state $(q_1,q_2)$ to a state in $F\times (Q\setminus F)$ or $(Q\setminus F) \times F$ 
	to find a word that leads to an accepting state when read from $q_1$ and a rejecting state when read from $q_2$ or vice versa.
	Since a shortest simple path in a graph is bounded by the number of nodes, the shortest such word is of length at most $|Q|^2$.
	The shortest path can be found using breadth-first search algorithms that run in time linear in the number of vertices and edges,
	thus polynomial in  the size of the DFA.
\end{proof}	

Since computing a spanning tree (in particular via DFS) and finding shortest paths can be done in polynomial time this shows that for DFAs we can construct the characteristic set in polynomial time. That is, while Definition~\ref{def:ident-limit} only requires that the characteristic set be of polynomial size, for DFAs we can show that it can also be computed in polynomial time.

\subsection{Inferring a DFA}\label{sec:inferDFA}
Next we describe algorithm $\InferDFA$ that given a sample of words $\aut{S}$, infers 
from it in polynomial time a DFA that agrees with  $\aut{S}$. And moreover, if $\aut{S}$ subsumes
the characteristic set $\aut{S}_\aut{D}$ of a DFA $\aut{D}$ then  $\InferDFA$ returns a DFA that recognizes $\aut{D}$. 

Let $W$ be the set of words in the given sample $\aut{S}$ (without their labels). 
Let $R$ be the set of prefixes of $W$ and $C$ the set of suffixes of $W$.
Note that  $\epsilon\in R$ and $\epsilon \in C$.
Let $r_0,r_1,\ldots$ be some enumeration of $R$ and $c_0,c_1,\ldots$ some enumeration of $C$ where $r_0=c_0=\epsilon$.
In the sequel we often use $i_w$ for the index of $w$ in $R$.
The algorithm builds a matrix $M$ of size $|R|\times|C|$ whose entries take values in $\{0,1,?\}$.
The algorithm set the value of entry $(i,j)$ as follows. If $r_i c_j$ is not in $W$, it is set to $?$. 
Otherwise it is set to $1$ iff the word $r_i c_j$ is labeled $1$ in $\mathcal{S}$. 
We get that 
$r_i \sim_\mathcal{S} r_j$ iff
for every $k$ such  that both $M(i,k)$ and $M(j,k)$ are different from $?$ we have that $M(i,k)=M(j,k)$.

The algorithm 
sets $R_0=\{\epsilon\}$. Once $R_i$ is constructed, the algorithm tries to establish whether $r\sigma$ for $r\in R_i$ and $\sigma\in\Sigma$
is distinguished from all words in $R_i$. It does so by considering all other words $r'\in R_i$ and checking whether $r \sim_\mathcal{S} r'$. 
If $r\sigma$ is found to be distinct from all words in $R_i$, then $R_{i+1}$ is set to $R_i\cup\{r\sigma\}$. The algorithm proceeds until no new words are distinguished. Let $k$ be minimal such that $R_k=R_{k'}$ for all $k'>k$, and let $R= R_k$.
 If not all words in $R$ are in $W$ (that is $M(i,0)= \ ?$ for some $r_i\in R$), the algorithm returns the prefix-tree automaton.\footnote{The prefix-tree automaton is the automaton obtained by placing all words in a tree data structure (sharing common prefixes) and labeling a state accepting iff the unique word reaching that state is in the sample and is labeled $1$.}
 Otherwise, the states of the constructed DFA are set to be the words in $R$. The initial state is $\epsilon$ and a state $r_i$ is classified as accepting 
iff $M(i,0)=1$ (recall that the entry $M(i,0)$ stands for the value of $r_i\cdot \epsilon$ in $\mathcal{S}$).
To determine the transitions, for every $r\in R$ and $\sigma \in \Sigma$, recall that there exists at least one state $r'\in R$ that cannot be distinguished from $r\sigma$.  
The algorithm then adds a transition from $r$ on $\sigma$ to $r'$.

\begin{prop} \quad
	\begin{enumerate}
		\item Algorithm $\InferDFA$ runs in polynomial time and returns a DFA that agrees with the given sample $\mathcal{S}$.
		\item 	Let $\mathcal{S}_\aut{D}$ be the sample constructed for a DFA $\aut{D}$ by algorithm $\CharDFA$, and let $\mathcal{S} \supseteq \mathcal{S}_\aut{D}$.
		Then algorithm $\InferDFA$ returns a DFA that recognized the same language as $\aut{D}$. 
	\end{enumerate}
\end{prop}

\begin{samepage}
\begin{proof} \hfill % Move proof to next page and start enumeration in a new line
		\begin{enumerate}
			\item The number of prefixes (or suffixes) of a set of words is bounded by the size of the longest word times the size of the set. Thus $M$ is of polynomial size, and so is its construction. The number of iterations required for converging the $R_i$ sets is bounded by $|W|$. The prefix-tree automaton can be computed in polynomial time.
			Determining acceptance is polynomial in $|R|$, and
			determining the transitions is polynomial in  $|R|\times|\Sigma|$. 
			Thus the overall running time of the algorithm is polynomial.
			
			Clearly, if the algorithm returns the prefix-tree automaton then it agrees with the given sample $\aut{S}$.
			We claim that it agrees with the given sample also in the second case.
			We show, by induction on the length of the word, that for every $w\in W$, 
			if $w$ reaches state $r$ of the constructed DFA, then $w\sim_\mathcal{S} r$.
			Since $w$ is in the sample, and $r$ is in the sample (otherwise the algorithm would return the prefix-tree automaton),
			it follows that $M(i_w,0)=M(i_r,0)$ hence the DFA agrees with the sample on $w$.
			
			For the base case we have that $|w|=0$ then the DFA accepts if $r_0$ is accepting, which holds iff $M(0,0)=1$. Indeed this entry is filled with the label of $\epsilon$ in $\mathcal{S}$.
			Consider now $w=v\sigma$ for some $v\in\Sigma^*$ and $\sigma\in\Sigma$.
			Assume the DFA reaches state  $s_{\ell}$ on reading $v$ and $s_m$ after reading $w$. 
			By induction hypothesis, we know that $r_\ell \sim_{\mathcal{S}} v$.
			From the construction of the algorithm it follows that $r_\ell \sigma \sim_{\mathcal{S}} s_m$ as otherwise, reading $\sigma$ from $r_\ell$ would lead to a different state. If $r_m \not  \sim_{\mathcal{S}} w$ then there exists a suffix $c_i\in C$ s.t. $M(m,i)\neq M(i_w,i)$. But
			then $\sigma c_i$ is also in $C$, denote it by $c_j$. Then  $M(\ell,j)\neq M(i_v,j)$   contradicting that 
			$r_\ell \sim_{\mathcal{S}} v$.
			
			\item
			Next we show that if $\mathcal{S}$ subsumes $\mathcal{S}_\aut{D}$ then the returned DFA agrees with $\aut{D}$.
			Let $w_1,\ldots,w_n$ be the set of accessible words chosen by $\CharDFA$. 
			Since $\mathcal{S}$ consists of a distinguished word for every pair of access words $w_i,w_j$ of $\aut{D}$, algorithm $\InferDFA$
			will determine $w_i \not\sim_\mathcal{S} w_j$ and $R$ will consist of at least $n$ states. It may not consist of more states, since
			the sample has to agree with the language of $\aut{D}$ and every word $w$ agrees with some state of $\aut{D}$ on all possible suffixes, thus $\InferDFA$ cannot determine that $w$ corresponds to a new distinct state. 
			%be determined distinct by $\InferDFA$.
			Since $S \cdot \Sigma \cdot E$ was placed in $\mathcal{S}$, for every distinguished state $w$ and every $\sigma\in\Sigma$ the algorithm $\InferDFA$ can determine the transition from $w$ upon reading $\sigma$.
			Since $S\cdot \epsilon$ is placed in $\mathcal{S}$, algorithm $\InferDFA$  can correctly label acceptance of states.
			Thus the obtained DFA is isomorphic to the original DFA. \qedhere
			\end{enumerate}
	\end{proof}
\end{samepage}

We can thus conclude that DFAs are identifiable in the limit using polynomial time and data. Furthermore, they satisfy
the properties of Theorem~\ref{prop:charDFA}.

\end{document}